\documentclass[11pt]{article}
\usepackage[utf8]{inputenc}
\usepackage{style}
\usepackage{mdframed}
\usepackage{amsmath}
\usepackage{mdwlist}
\usepackage{nicefrac}
\usepackage[parfill]{parskip}
\usepackage{setspace}
\usepackage{hyperref}
\usepackage[ruled,linesnumbered]{algorithm2e}

\newcommand{\eps}{\varepsilon}
\newcommand{\slc}{\ensuremath{\textsc{Smooth-}2k\textsc{-Label Cover}}}

\newcommand{\Ex}{\E}

\newcommand{\Langle}{\left\langle}
\newcommand{\Rangle}{\right\rangle}
\newcommand{\bX}{{\bf X}}
\newcommand{\bY}{{\bf Y}}
\newcommand{\bc}{{\bf c}}

\newcommand{\mc}[1]{\ensuremath{\mathcal{#1}}\xspace}
\newcommand{\mb}[1]{\ensuremath{\mathbf{#1}}\xspace}
\newcommand{\tn}[1]{\ensuremath{\textnormal{#1}}\xspace}
\newcommand{\ol}[1]{\ensuremath{\overline{#1}}\xspace}
\newcommand{\wh}[1]{\ensuremath{\widehat{#1}}\xspace}
\allowdisplaybreaks

\title{Hardness of Learning DNFs using Halfspaces}
\author{Suprovat Ghoshal\thanks{Indian Institute of Science. E-mail: \texttt{suprovat@iisc.ac.in}} \and Rishi Saket\thanks{IBM Research, Bangalore. E-mail: \texttt{rissaket@in.ibm.com}}}
\begin{document}
\date{}
\maketitle
\thispagestyle{empty}

\begin{abstract}
	The problem of learning $t$-term DNF formulas (for $t = O(1)$) has been studied extensively in the PAC model since its introduction by Valiant (STOC 1984). A $t$-term DNF can be efficiently learnt using a $t$-term DNF only if $t = 1$ i.e., when it is an {\sf AND}, while even \emph{weakly} learning a $2$-term DNF using a constant term DNF was shown to be \NP-hard by Khot and Saket (FOCS 2008). On the other hand, Feldman et al. (FOCS 2009) showed the hardness of weakly learning a noisy {\sf AND} using a \emph{halfspace} -- the latter being a generalization of an {\sf AND}, while Khot and Saket (STOC 2008) showed that an intersection of two halfspaces is hard to weakly learn using any function of constantly many halfspaces. The question of whether a $2$-term DNF is efficiently learnable using $2$ or constantly many halfspaces remained open.
	In this work we answer this question in the negative by showing the hardness of weakly learning a $2$-term DNF as well as a noisy {\sf AND} using any function of a constant number of halfspaces. In particular we prove the following. 

\medskip
\noindent
	For any constants $\nu, \zeta > 0$ and $\ell \in \mathbbm{N}$, given a distribution over point-value pairs $\{0,1\}^n \times \{0,1\}$, it is \NP-hard to decide whether,

	\begin{itemize}
		\item
	\tn{YES Case.} There is a $2$-term DNF that classifies all the points of the distribution, and an {\sf AND} that classifies at least $1-\zeta$ fraction of the points correctly.
		
\item
	\tn{NO Case.} Any boolean function depending on at most $\ell$ halfspaces classifies at most $1/2 + \nu$ fraction of the points of the distribution correctly.
	\end{itemize}

	\medskip
	\noindent
	Our result generalizes and strengthens the previous best results mentioned above on the hardness of learning a $2$-term DNF, learning an intersection of two halfspaces, and learning a noisy {\sf AND}.

\end{abstract}
\thispagestyle{empty}

\newpage

\setcounter{page}{1}

\section{Introduction}

A boolean function $f: \{0,1\}^n \mapsto \{0,1\}$ can always be represented in a \emph{disjunctive normal form} (DNF) consisting of an {\sf OR} over {\sf AND}s of boolean literals. DNFs are widely studied in several fields in computer science such as complexity theory~\cite{SBI04}~\cite{Alekh05}\cite{AMO15}, learning theory~\cite{Kliv10}~\cite{Lee10}, pseudorandomness~\cite{Baz09}~\cite{GRR13} and combinatorics~\cite{XLZ18}~\cite{LZ19} to name a few. In many scenarios, such as circuit design or neural network training, the \emph{size} - given by the number of terms -- of a DNF is an important consideration, and much work has studied the properties of small DNFs. Contributing to this line of research, our work investigates the learnability of bounded term DNFs in the \emph{probably approximate correct} (PAC) model~\cite{Val84}.

Formally, a \emph{concept class} $\mc{C}$ of boolean functions is efficiently learnable if, given access to random samples from a distribution over point-value pairs realizable by an (unknown) member of $\mc{C}$, there exists a polynomial time randomized algorithm which outputs with high probability a hypothesis (from a \emph{hypothesis class} \mc{H})  which is consistent with the distribution's samples with probability (accuracy) arbitrarily close to $1$. A few useful variants are \emph{proper learning} when the hypothesis is restricted to be from the concept class ($\mc{H} = \mc{C}$), and  
 \emph{weak learning} wherein the accuracy of learning is allowed to be any constant greater than the trivially achievable random threshold of half.

The question of efficient learnability of DNFs was raised in Valiant's celebrated work~\cite{Val84} which introduced PAC learning. While a $1$-term DNF (i.e., an {\sf AND}), is properly learnable, it is $\NP$-hard to properly learn a $t$-term DNF~\cite{VP88}, for any $t \geq 2$. This is not true when allowed more general hypotheses (improper learning): Valiant~\cite{Val84} showed that a $t$-term DNF can be efficiently learnt using a \emph{conjunctive normal form} (CNF) which is an {\sf AND} of {\sf OR}s (clauses), where each clause has at most $t$ literals. 

For learning unbounded term DNFs, Bshouty~\cite{bshouty96} first gave a $2^{O(\sqrt{n \log t})}$-time algorithm which was later improved by Klivans and Servedio~\cite{KS04}. They showed that any DNF can be expressed as a polynomial threshold of degree $O(n^{1/3}\log t)$, and therefore can be learnt using linear programming in time $n^{\tilde{O}(n^{1/3})}$ which is state of the art for learning unrestricted DNF.

On the complexity side, assuming $\NP \not\subseteq \ZPP$, Nock, Jappy and Salantin~\cite{NJS98} showed that any $n^{c}$-sized DNFs cannot efficiently be learnt using $n^{\gamma c + \eta}$-term DNFs, for certain ranges $\gamma,\eta$, which was later strengthened by Alekhnovich et al.~\cite{Alekh08} to any constant value of $\gamma$, assuming $\NP \neq \RP$ thus ruling out proper learning of arbitrary DNFs in polynomial time. Subsequently, Feldman~\cite{Fel09} showed this to hold even with access to membership queries. While these results do not rule out properly learning restricted DNFs, Alekhnovich et al.~\cite{Alekh08} showed that it is \NP-hard to learn a $2$-term DNFs using a $t$-term DNF for any constant $t$. Subsequently, this was strengthened by Khot and Saket~\cite{KS08b} who showed the \NP-hardness of even weakly  
learning $2$-term DNFs using constant term DNFs. In related work, Feldman~\cite{Fel06} and Feldman et al.~\cite{FGKP09} showed the hardness of weak agnostic learning \emph{noisy} AND with an AND i.e., in the presence of a small fraction of adversarially perturbed sample labels, whereas Khot and Saket~\cite{KS08b} showed the same when using a {\sf CNF} with bounded clause width as hypothesis\footnote{This also rules out as hypothesis {\sf AND} of $O(1)$-arity functions, as the latter can be represented as $O(1)$-width {\sf CNF}s.}, complementing Valiant's~\cite{Val84} algorithmic work mentioned above in the non-noise setting.

Linear threshold functions (LTFs), a.k.a. \emph{halfspaces}, given by $\{{\bf x} : \langle {\bf c}, {\bf x}\rangle + \theta > 0\}$, are a natural and well studied class of hypothesis used in machine learning, especially as the output of linear kernel models such as Perceptron and SVMs. They generalize DNFs: a $t$-term DNF is can be represented as an {\sf OR} of $t$ halfspaces. While Valiant's work~\cite{Val84} showed that a $t$-term DNF can be learnt using an intersection of \emph{unbounded} number of halfspaces, using linear programming a $1$-term DNF can be efficiently learnt using one halfspace. A natural question is whether the latter holds for $t = 2$ as well:
\begin{equation}
	\tn{Can a $2$-term DNF be efficiently learnt using two halfspaces?} \label{eqn:ques} 
\end{equation}
In this context Khot and Saket~\cite{KS08} showed that an intersection of two halfspaces -- which captures (the negation of) a $2$-term DNF -- cannot be efficiently weakly learnt using any function of constant number of halfspaces, unless $\NP \neq \RP$. Feldman et al.~\cite{FGRW12} showed the hardness of weakly learning a noisy {\sf AND} using a single halfspace, generalizing previous works of Guruswami and Raghavendra~\cite{GR09} and Feldman et al.~\cite{FGKP09}, which proved the same hardness for learning a noisy halfspace, and results of \cite{Fel06, FGKP09} on the hardness of learning a noisy {\sf AND} with an {\sf AND}.

In our work we answer \eqref{eqn:ques} in the negative, proving the \NP-hardness of learning $2$-term DNF (as well as noisy {\sf AND}) using functions of any constant number of halfspaces. In particular we prove:

\begin{theorem} \label{thm:mainyesno}
	For any constants $\nu, \zeta > 0$ and $\ell \in \mathbbm{N}$, the following holds. Given a distribution\footnote{The distribution in our reduction is explicitly described by a polynomial sized set of point-label pairs.} $\mc{D}$ over $\{0,1\}^n \times \{0,1\}$, it is \NP-hard to distinguish between the two following cases:
	\begin{itemize}
		\item \tn{YES Case:} There exist two {\sf AND}s $\mc{A}_1$ and $\mc{A}_2$ such that
\medskip

			\noindent
			\begin{minipage}{.5\linewidth}
			\begin{equation}
			\Pr_{(\mb{x}, a) \gets \mc{D}} \left[\left(\mc{A}_1(\mb{x}) \vee \mc{A}_2(\mb{x})\right) = a\right] = 1, \label{eqn:YES1} 
			\end{equation}	
			\end{minipage}
			\begin{minipage}{.5\linewidth}
			\begin{equation}
			\Pr_{(\mb{x}, a) \gets \mc{D}} \left[\mc{A}_1(\mb{x}) = a\right] \geq 1 - \zeta.\label{eqn:YES2}
			\end{equation}
			\end{minipage}
			\medskip

		\item \tn{NO Case:} Any boolean function $f$ depending on at most $\ell$ halfspaces satisfies, 
			$$\Pr_{(\mb{x}, a) \gets \mc{D}} \left[f(\mb{x}) = a\right] \leq \frac{1}{2} + \nu.$$
		
	\end{itemize}

\end{theorem} 

Our result strengthens and generalizes the hitherto best results on the hardness of learning (i) constant term DNFs using constant term DNFs,  (ii) noisy {\sf AND}s using halfspaces, and (iii) intersections of two halfspaces using constantly many halfspaces, given in the works of \cite{KS08b}, \cite{FGRW12} and \cite{KS08}. An interesting feature of our result is that the hardness results for $2$-term DNF as well as noisy {\sf AND} and \emph{simultaneously} shown through the same hard instance satisfying \eqref{eqn:YES1}  and \eqref{eqn:YES2} in Theorem \ref{thm:mainyesno}. We note that our result regarding $2$-term DNF is tight in terms of the degree of the hypothesis threshold function: any $2$-term DNF can be efficiently learnt using a single quadratic threshold. More generally, by representing the negation of a $t$-term DNF as the sign of a product of $t$ linear forms, a $t$-term {\sf DNF} is learnable using a degree-$t$ \emph{polynomial threshold function} (PTF) in time $n^{O(t)}$.

\subsection{Other Related Work}

The problem of learning DNFs has also been studied in distributional settings. In particular, Verbeugt~\cite{Ver90} gave a quasi-polynomial time algorithm for learning DNFs under the uniform distribution on the hypercube. Blum et al.~\cite{blum94} gave a polynomial time algorithm for weakly learning DNFs under the uniform distribution using membership queries. This was subsequently improved on by Jackson~\cite{jackson97} who gave a polynomial time algorithm for strongly learning DNFs in the same setting. There have been several other works which have studied the problem of learning DNFs under a wide range of alternative restricted settings \cite{BMDS03}\cite{Serv04}\cite{KST09}\cite{Feld12}.

There has been extensive work on algorithms which learn concepts using halfspaces as the output hypothesis class such as Perceptron~\cite{Rosen58}, SVMs~\cite{Vap98} etc., as well as on algorithmic results for learning more general concepts. Klivans et al.~\cite{KDS04} showed that one can learn arbitrary functions of constant number of halfspaces in quasi-polynomial time under the uniform distribution. Diakonikolas et al.~\cite{DHKRST10} gave an algorithm that learns  degree-$d$ PTFs in time $2^{O_\epsilon(d^2)}$ under the uniform distribution over the hypercube. Recently Gottlieb et at.~\cite{GKKN18} gave an algorithm for PAC learning intersections of $t$-halfspaces with intersections of $t\log t$ halfspaces with margin $\gamma$ in time $n^{O(1/\gamma^2)}{\rm poly}(t)$. On the hardness side, Diakonikolas et al.~\cite{DOSW11} improved on the works of \cite{GR09,FGKP09} showing the intractability of weakly learning a noisy halfspace using degree-$2$ PTFs, which was generalized to rule out all constant degree PTFs as hypotheses by Bhattacharyya et al.~\cite{BGS18}. 

There are also hypothesis independent hardness of learning results for learning (relatively) small DNFs. In \cite{daniely16}, Daniely showed that unless there exist efficient algorithm for strongly refuting $K$-XOR formulas for certain clause densities, there does not exist any polynomial time algorithm for learning $t$-term DNFs (with $t = \omega(\log n)$), using any polynomial time evaluatable hypothesis class. Earlier, Klivans and Sherstov~\cite{KS09} had shown that there is no efficient algorithm for learning intersection of $n^\epsilon$-halfspaces using unrestricted hypothesis classes, under different cryptographic hardness assumptions. Such a result is not possible for $O(1)$-term DNF as it is efficiently learnable using a {\sf CNF} (by \cite{Val84}) and by an $O(1)$-degree PTF as observed earlier in this section. 

On this note, we remark that the result of Applebaum, Barak and Xiao~\cite{ABX08} shows that hypothesis independent hardness of learning results assuming $\P \neq \NP$ are unlikely without making significant breakthroughs in complexity theory. Therefore, any complexity theoretic hardness of learning result for noisy {\sf AND} will probably require some restriction of the hypothesis. While our result rules out functions of constantly many halfspaces as hypotheses, it remains an important open problem to show the same for (functions of) polynomial thresholds.

\subsection{Overview of the Reduction}

Our reduction is from a variant of Label Cover and uses the standard template of a suitably constructed \emph{dictatorship test} defined over a set of coordinates. This test gives a distribution over point-value pairs when applied to the coordinates corresponding to each edge (or a local collections of edges) of the Label Cover, and the union of these distributions constitutes the instance of the learning problem. Roughly speaking, the distribution for an edge should satisfy (i) (completeness) any matching labeling should yield a good classifier from the concept class, and (ii) (soundness) any good enough hypothesis should yield a matching labeling to the vertices of the edge.  For ease of exposition in this section, we shall not elaborate on our hardness for learning noisy {\sf AND}, and focus on the $2$-term DNF concept class. Furthermore, instead of $2$-term DNF we shall consider $2$-clause CNF (these are essentially equivalent: negating one yields the other).

We begin by describing in Figure \ref{fig:basic-test-0} an elementary dictatorship test -- a generalized version of which was used by \cite{KS08b} to rule out weakly learning a $2$-clause CNF using constant clause CNF (equivalently a $2$-term DNF with constant term DNF). Let us denote by $\mc{I}$ the generated distribution, and by $\mc{I}_0$ and $\mc{I}_1$ its restrictions to the $0$ and $1$ valued points respectively. 

\begin{figure}[h!]
	\begin{mdframed}
		{\bf Coordinate Space}: Consists of variables $\bX = \{X_i\}^{M}_{i = 1}$ and $\bY_{r} = \{Y_{r, i}\}^{M}_{i =1}$, for each $r \in [k]$. 
		
		{\bf Test}: Sample the  $a \overset{\rm u.a.r.}{\sim} \{0,1\}$, and sample an $a$-labeled point as follows. \\
		Sample $s \overset{\rm u.a.r.}{\sim} [k]$. 
		For every $i \in [M]$,.
		\begin{enumerate}
			\item If $a = 1$, set $X_i = Y_{s, i} = 1$. 
			\item If $a = 0$, then sample $r \overset{\rm u.a.r.}{\sim} [k]$. W.p. $1/2$ set $X_i = 1,Y_{r,i} = 0$ and w.p. $1/2$ set $X_i = 0,Y_{r,i} = 1$.
			\item Set the rest of the variables to $0$.
		\end{enumerate}
		Output $(\bX, \bY)$ with label $a$.
	\end{mdframed}
	\caption{Distribution $\mc{I}$}
	\label{fig:basic-test-0}
\end{figure}

For the completeness, observe that for every $i \in [M]$, the CNF:\  $X_i \wedge (\vee_{r\in [k]}Y_{k, i})$ classifies all the points of $\mc{D}$. On the other hand, it can be shown that any $t$-clause CNF that classifies $1/2 + \Omega(1)$ fraction of the points correctly yields two different $\bY_{r}$ and $\bY_{r'}$  that can be \emph{decoded} into a common element of $[M]$. The analysis in \cite{KS08b} used fairly simple structural and probabilistic arguments and we refer the reader to it for the detailed proof. 

Unfortunately, it is easy to see that the above distribution fails with halfspaces as hypotheses because of the following simple observation: the first moment of each of the variables $\{X_i\}^{M}_{i = 1}$ under $\mc{I}_1$ is $1$, while it is $1/2$ under $\mc{I}_0$. Since  $\{X_i\}^{M}_{i = 1}$ are independent Bernoulli variables in both cases, by straightforward concentration this difference in first moments can be leveraged by a non-dictatorial halfspace (whose linear form simply sums up the $X_i$s) to classify the points with high probability.

Another issue that impairs $\mc{I}$ from working for halfspaces is the lack of \emph{noise}. Consider a linear form whose coefficient for each $X_i$ is $2^i$ and for each $\{Y_{r, i}\}^{k}_{r =1}$ is $(-2^i)$. Such a linear form always evaluates to $0$ on $\mc{I}_1$ while it is always non-zero on a $\mc{I}_0$, and therefore this instance admits a good classifier using an intersection of two halfspces. This pathological example also illustrates the complications arising out of coefficient vectors that are \emph{not} \emph{regular} i.e., they contain large sequences of geometrically decreasing coefficients.  

Our test -- represented  by the distribution $\mc{D}$ and its restrictions $\mc{D}_0$ and $\mc{D}_1$ -- is designed to fix the above shortcomings. Firstly, we ensure by having $k$ blocks of the $\bX$ variables as well, that the individual marginal distribution of any variable is the same for both, $\mc{D}_0$ and $\mc{D}_1$. In order to introduce noise, a large common randomly chosen set of variables are sampled i.i.d Bernoulli, under both $\mc{D}_0$ and $\mc{D}_1$. Restricted to each $i \in [M]$, this set resides either in $\bX$ or $\bY$ depending on which side holds non-zero variables under $\mc{D}_0$. Furthermore, the test also incorporates Label Cover projections $[M] \to [m]$, so that all the variables corresponding to those labels projecting to some $j \in [m]$ are sampled together, and independently of those  projecting to $j' \neq j$. In order to bound the variance of the difference of samples from $\mc{D}_0$ and $\mc{D}_1$, after fixing the common noise set, we replace each variable $Z$ with a collection of $Q$ variables $\{Z_q\}_{q=1}^Q$, sampling them u.a.r. from $\{0,1\}^Q$ if the $Z$ was part of the noise set, u.a.r. from\footnote{Here ${\bf e}_q$ is the $Q$-length indicator vector of the $q$th coordinate, $q\in [Q]$.} $\{{\bf e}_1, \dots, {\bf e}_Q\}$  if $Z$ contributed to the difference between $\mc{D}_0$ and $\mc{D}_1$, and are each set to $0$ if $Z$ was set to $0$. The parameter $Q$ is taken to be much larger than the pre-image size of the Label Cover projections.

The analysis of a coefficient vector utilizes the notion of  its \emph{critical index} which was introduced by Servedio~\cite{Serv06}, and also used in the above mentioned work of \cite{FGRW12}. As part of their analysis, \cite{FGRW12} proved the invariance of halfspaces with regular coefficients under distributions with matching first to third moments. They apply this to $\mc{D}_0$ and $\mc{D}_1$ ensuring that their appropriate moments are matched. In our case however, this is impossible for points classified by a $2$-clause CNF, since the product of the two linear forms representing the CNF is always zero under $\mc{D}_0$, and positive under $\mc{D}_1$. The invariance shown by \cite{FGRW12} only bounded the deviation of the biases of halfspaces under $\mc{D}_0$ and $\mc{D}_1$, and thus their analysis only rules out single halfspaces as hypothesis. In contrast, our work bounds the expected point-wise deviation allowing us to extend our result to  functions of constantly many halfspaces. For technical considerations, as in \cite{FGRW12} we combine our test with \emph{Smooth} Label Cover whose (smoothness) property guarantees that the projections for most its edges behave as bijections when restricted to fixed small sets of labels for its vertices.

At a high level our analysis proceeds as follows: we first truncate the coefficient vector to a small irregular part. Subsequently, our analysis bounds the variance of the difference between $\mc{D}_0$ and $\mc{D}_1$ given a noise set which we show picks up enough mass with high probability. Finally, we directly apply anti-concentration on the noise variables to bound the deviation between $\mc{D}_0$ and $\mc{D}_1$. In the rest of this section we present a simplified version of the our test distribution, and informally describe its analysis. A slightly modified version (with differently biased random choices) yields the desired hardness result for noisy {\sf AND} as well.

\subsection{Simplified Distribution and Sketch of Analysis}
Consider the distribution $\mc{D}$ given in Figure \ref{fig:basic-test-2}.
The completeness case is straightforward. Indeed, fix any $i^* \in [M]$ and the following $2$-clause CNF: $f^*({\bf X}, {\bf Y}) := \left(\bigvee_{r \in [k]}\bigvee_{q \in [Q]} X_{r,i^*,q} \right) \wedge \left(\bigvee_{r \in [k]}\bigvee_{q \in [Q]} Y_{r,i^*,q} \right)$. Let $j^*$ be such that $i \in B_{j^*}$. If $a = 0$, then  $b_{j^*} = 0$ which implies ${\bf Y}^{(j^*)} = {\bf 0}$, and $b_{j^*} = 1$ which implies ${\bf X}^{(j^*)} = {\bf 0}$. Therefore, either $\bigvee_{r \in [k]}\bigvee_{q \in [Q]} X_{r,i^*,q}  = 0$, or $\bigvee_{r \in [k]}\bigvee_{q \in [Q]} Y_{r,i^*,q} = 0$, which means that $f^*  = 0$. On the other hand if $a = 1$, then there are $r, r'$ s.t. ${\bf X}_{r, i^*}, {\bf Y}_{r',i^*} \in   \{{\bf e}_1,{\bf e}_2,\ldots,{\bf e}_Q\}$, implying that $f^* = 1$.
\begin{figure}[h!]
	\begin{mdframed}
		{\bf Coordinate Space}: Given a partition of $[M]$ into blocks $B_1,B_2,\ldots,B_m$, each of size $d$. For each $j \in [m]$, introduce a set of $dk$ vector-variables $\bX^{(j)} :=\bigoplus_{r \in [k],i \in B_j}\bX_{r,i}$ and $\bY^{(j)} := \bigoplus_{r  \in [k], i \in B_j}\bY_{r,i}$. For each $r \in [k]$ and $i \in B_j$, let $\bX_{r,i} = \left\{X_{r,i,q}\right\}_{q \in [Q]}$ and $\bY_{r,i} = \left\{Y_{r,i,q}\right\}_{q \in [Q]}$ 
		
		{\bf Test}: Sample $a \sim \{0,1\}$ uniformly at random to denote the value of the output point. \\
		For every $j \in [m]$ independently do the following.
		\begin{enumerate}
			\item Sample $b_j  \in \{0,1\}$ uniformly at random.
			\item Sample a $k/2$ ($k$ is chosen to be even) sized subset $S_j \overset{\rm u.a.r}{\sim} {[k] \choose {k/2}}$.
			\item If $b_j = 0$, for every $r \not\in S_j$, $i \in B_j$ sample ${\bf X}_{r,i} \overset{\rm u.a.r}{\sim} \{0,1\}^Q$. Otherwise, if $b_j = 1$, for every $r \not\in S_j$, $i \in B_j$  sample ${\bf Y}_{r,i} \overset{\rm u.a.r}{\sim}\{0,1\}^Q$.
			\item Independently sample $r,r'\overset{\rm u.a.r}{\sim}[k]\setminus S_j$.
			\item If $a = 0$, then with probability $1/k$ do the following:  (i) if $b_j = 0$, then for every $i \in S_j$, sample $\bX_{r,i} \overset{\rm u.a.r}{\sim} \{{\bf e}_1,{\bf e}_2,\ldots,{\bf e}_Q\}$, otherwise (ii) if $b_j = 1$, for every $i \in S_j$, sample $\bY_{r,i} \overset{\rm u.a.r}{\sim} \{{\bf e}_1,{\bf e}_2,\ldots,{\bf e}_Q\}$.
			\item If $a = 1$, then for every $i \in B_j$ sample $\bX_{r,i},\bY_{r',i} \overset{\rm u.a.r}{\sim} \{{\bf e}_1,{\bf e}_2,\ldots,{\bf e}_Q\}$.
		\end{enumerate}

			Set the rest of the variables to $0$. Output the point $({\bf X}, {\bf Y})$ with label $a$.
			\end{mdframed}
	\caption{Distribution $\mc{D}$}
	\label{fig:basic-test-2}
\end{figure}

For the soundness assume that there is some function $f$ of $\ell$ halfspaces ${\rm pos}(h_s({\bf X}, {\bf Y}))$, $s = 1, \dots, \ell$ that classifies $1/2 + \nu$ fraction of the points of $\mc{D}$ correctly, where ${\rm pos}(\cdot)$ is the sign function. Using the definitions related to the critical index~\cite{Serv06} (see Section \ref{sec:criticalindex}) we define the sets $C_{X, s, r}$, $C_{X, s,r}^{\leq K}$, and $I_{X, s,r}$, letting ${{\bf c}}$ to be the coefficient vector of $h_s$ corresponding to the variables ${\bf X}_r$, as follows:
\begin{itemize}
	\item $C_{X, s, r}$ is the set $C_\tau({{\bf c}})$, and $C_{X, s,r}^{\leq K} \subseteq C_{X, s, r}$ are the top $K$ elements of $C_{X, s, r}$ given by $C^{\leq K}_\tau({{\bf c}})$
	\item  $I_{X, s,r} := C^{\leq K}_{X, s,r} \cup \{i \in [M]\setminus C_{X, s,r} : \|{{\bf c}}_i\|^2 > (1/d^8)\sum_{i' \in [M]\setminus C_{X, s,r} }\|{{\bf c}}_{i'}\|^2\}$,
\end{itemize}
and similarly $C_{Y, s, r}$, $C_{Y, s,r}^{\leq K}$, and $I_{Y, s,r}$, for $s \in [\ell], r \in [k]$, and some parameters $\tau$ and $K$ that we choose appropriately. Clearly the size of any $I_{X, s,r}$ or $I_{Y, s,r}$ is at most $2(K + d^8)$. Using the smoothness property of the Label Cover instance we can assume that for each $r$
\begin{equation}
	\left|\left(\bigcup_{s\in[\ell]}I_{X, s,r}\right) \cap B_j \right| \leq 1, \tn{\ \ and \ \ } \left|\left(\bigcup_{s\in[\ell]}I_{Y, s,r}\right) \cap B_j \right| \leq 1 \tn{\ \ for all\ }j  \in [m]. \label{eqn:uniqueproj}
\end{equation}
Further we may assume that no two subsets from $\{C^{\leq K}_{X, s, r}, C^{\leq K}_{Y, s, r} : s \in [\ell], r\in [k]\}$ can have indices from the same block $B_j$. Otherwise using \eqref{eqn:uniqueproj}, there are two distinct tuples $(Z_1, r)$ and $(Z_2, r')$  where $Z_1, Z_2 \in \{X, Y\}$ such that $C^{\leq K}_{Z_1, s, r}$ and $C^{\leq K}_{Z_2, s', r'}$ share indices from a common block. This yields a good labeling. This, along with \eqref{eqn:uniqueproj} implies that any $B_j$ can have at most one element from $\bigcup_{s\in[\ell], r\in[k]}(C^{\leq K}_{X, s, r}\cup C^{\leq K}_{Y, s, r})$. 

{\bf Pairing Distribution.} 
Using the description of $\mc{D}$ along with the above observations one can construct a distribution $\wh{\mc{D}}$ over $(({\bf X}^0, {\bf Y}^0),  ({\bf X}^1, {\bf Y}^1))$ such that its marginals are $\mc{D}_0$ and $\mc{D}_1$, as follows. 
\begin{itemize}
	\item Sample the subsets $\{S_j\}_{j\in [m]}$, and for all $j \in [m]$ such that $B_j$ contains some element (at most one from above) from $\bigcup_{s\in[\ell], r\in[k]}(C^{\leq K}_{X, s, r}\cup C^{\leq K}_{Y, s, r})$, sample $b_j$. 
	\item Sample the values of all  ${\bf X}_{r, i}$ s.t. $i \in C^{\leq K}_{X, s, r}$ for some $s \in [\ell]$, and all ${\bf Y}_{r, i}$ s.t. $i \in C^{\leq K}_{Y, s, r}$ for some $s \in [\ell]$. Now sample the rest of the $b_j$s. 
	\item Sample values of the rest of (\emph{noise variables}) ${\bf X}_{r,i} \overset{\rm u.a.r}{\sim} \{0,1\}^Q$ s.t. $r \not\in S_j, b_j = 0$ and  ${\bf Y}_{r,i} \overset{\rm u.a.r}{\sim} \{0,1\}^Q$ s.t. $r \not\in S_j, b_j = 1$.
	\item Finally, conditioned on all these fixings, sample $({\bf X}^0, {\bf Y}^0)\gets \mc{D}_0$ and $({\bf X}^1, {\bf Y}^1)\gets \mc{D}_1$. 
\end{itemize}
		Note that this distribution is independent of the choice of any specific halfspace ${\rm pos}(h_s({\bf X}, {\bf Y}))$, which is crucial to the argument. From the goodness of the classifier $f$ and an averaging argument we obtain that there exists one out of the $\ell$ halfspaces, ${\rm pos}(h_{s^*}({\bf X}, {\bf Y}))$ such that,
\begin{equation}
	\E_{\wh{\mc{D}}}\left[\left|{\rm pos}(h_{s^*}({\bf X}^1, {\bf Y}^1)) - {\rm pos}(h_{s^*}({\bf X}^0, {\bf Y}^0))\right|\right] \geq \nu/2\ell. \label{eqn:pointwisedev1}
\end{equation}

{\bf Truncation Step.} This step uses the property that squared mass corresponding to the indices in $C_{X, s^*,r}^{\leq K}$ is much larger than that in $C_{X, s^*,r}\setminus C_{X, s^*,r}^{\leq K}$ and similarly for $C_{Y, s^*,r}^{\leq K}$ and $C_{Y, s^*,r}$, $r \in [k]$. Using this we show that in the coefficient vector of $h_{s^*}$ corresponding to  some ${\bf X}_r$ can be truncated by zeroing out those coefficients corresponding to $C_{X, s^*,r}\setminus C_{X, s^*,r}^{\leq K}$, while not disturbing \eqref{eqn:pointwisedev1} appreciably. This is obtained by bounding the deviation in the value of $h_{s^*}$ due to this zeroing out, and showing that this is overwhelmed by the anti-concentration of the noisy variables corresponding to $C_{X, s^*,r}^{\leq K/4}$. Our anti-concentration bound leverages the \emph{Littlewood-Offord-Erd\H os} lemma (see Lemma~\ref{lem:LO}), as well as the standard Berry-Esseen bound. 
Doing this truncation for coefficients corresponding to  each ${\bf X}_{r}$ and ${\bf Y}_{r}$, we obtain a linear form $h$ which satisfies \eqref{eqn:pointwisedev1} with $\nu/4\ell$ on the RHS. 
Let the coefficients of $h$ be given by ${\bf c}_X$ and ${\bf c}_Y$ corresponding to ${\bf X}$ and ${\bf Y}$ respectively. Note that truncation implies, $C^{\leq K}_\tau({\bf c}_{X, r}) =  C_\tau({\bf c}_{X, r})$ and $C^{\leq K}_\tau({\bf c}_{Y, r}) =  C_\tau({\bf c}_{Y, r})$, for $r \in [k]$. Call the blocks $B_j$ which are disjoint from all $C_\tau({\bf c}_{X, r})$,  $C_\tau({\bf c}_{Y, r})$ as \emph{regular} blocks. 

{\bf Structural Lemma.} In this we show that unless the truncated coefficient vectors of $h$ satisfy a certain structural property (which we leave unstated in this sketch), the LHS of \eqref{eqn:pointwisedev1} for $h$ is very small, leading to a contradiction. On the other hand, this property leads to a good labeling. The main idea is to show that (assuming this property is not satisfied) the truncation of $({\bf c}_X, {\bf c}_Y)$ implies with high probability over $\wh{\mc{D}}$ that the squared mass corresponding to the noise variables in the regular blocks is a significant fraction of the total squared mass in the regular blocks. The regularity of the blocks is leveraged to apply the Chernoff-Hoeffding inequality to obtain this. On the other hand, the variance of the difference of $h$ between $({\bf X}^0, {\bf Y}^0)$ and $({\bf X}^1, {\bf Y}^1)$ is nearly always much smaller than the squared mass of the noise variables. This uses the fact that our choice of $Q$ is large (in fact we need to set it to $\tn{poly}(d)$). Another application of Berry-Esseen along with Chebyshev's inequality completes the proof.

\subsection{Organization} 
Section \ref{sec:prelim} provides the necessary technical preliminaries required for the paper and describes the Smooth Label Cover problem. In Section \ref{sec:main-red} we give the theorem which states the guarantees of our reduction from Smooth Label Cover, followed by a description of the reduction. We prove the completeness and soundness of the reduction in Sections \ref{sec:completeness} and \ref{sec:soundness}. In Section \ref{sec:trun-crit-index} we prove  the guarantee of the truncation step, and in Section \ref{sec:struct} we prove the main structural lemma used in our soundness analysis. Finally in Appendix \ref{sec:LO} we prove the anti-concentration bound used in the truncation step, using the Littlewood-Offord-Erd\H os lemma as well as the Berry-Esseen theorem.

\section{Preliminaries}				\label{sec:prelim}

Throughout the paper we use ${\rm pos}(x) = \mathbbm{1}\{x \geq 0\}$ to denote the sign function. The following is a well known quantitative form of the central limit theorem.
\begin{theorem}[Berry-Esseen Theorem~\cite{DDon14}]			\label{thm:berry-ess}
	Let $X_1,\ldots,X_n$ be independent random variables with $\Ex[X_i] = 0$ and ${\rm Var}(X_i) = \sigma^2_i$, and assume that $\sum_{i \in [n]} \sigma^2_i = 1$. Let $\gamma := \sum_{i \in [n]} \|X_i\|^3_{L_3}$. Then
	\begin{equation}
		\sup_{t \in \mathbbm{R}} \left|\Pr_{X_1,\ldots,X_n}\left[\sum_{i \in [n]} X_i \leq t\right] - \Phi(t)\right| \leq c\gamma \label{eqn:BerryEsseen}
	\end{equation} 
	where $c$ is a universal constant, and $\Phi(\cdot)$ is the standard Gaussian CDF. Note that it follows from  above that \eqref{eqn:BerryEsseen} also holds when taking $\gamma := \max_{i \in [n]} (\|X_i\|^3_{L_3}/\sigma^2_i)$. 
\end{theorem}
Apart from the above, our anti-concentration bounds also use the following classical result for Bernoulli sums.
\begin{lemma}[Littlewood-Offord-Erd\H os Lemma~\cite{Erdos45}]			\label{lem:LO}
	Let $X_1,X_2,\ldots,X_n$ be i.i.d $\{0,1\}$-Bernoulli Random variables, and let $a_1,a_2,\ldots,a_n \in \mathbbm{R}$ such that for all $i \in [n]$ we have $|a_i| \geq 1$. Then there exists a constant $C> 0$ such that  
	\begin{align*}
	\Pr_{X_1,\ldots,X_n}\left[\left|\sum_{i \in [n]} a_i X_i + \theta \right| \leq 1 \right] \leq \frac{C}{\sqrt{n}}
	\end{align*}
	for any constant $\theta$.
\end{lemma}
We also use the following well known bounds of Chernoff-Hoeffding and Chebyshev.
\begin{theorem}[Chernoff-Hoeffding]\label{thm:chernoff}
 Let $X_1,\dots, X_n$ be independent random variables, each bounded as $a_i \leq X_i \leq b_i$ with $\Delta_i = b_i - a_i$ for $i = 1,\dots, n$. Then, for any $t > 0$,
	$$\Pr\left[\left|\sum_{i=1}^n X_i - \sum_{i=1}^n\E[X_i]\right| > t\right] \leq 2\cdot\tn{exp}\left(-\frac{2t^2}{\sum_{i=1}^n\Delta_i^2}\right).$$
\end{theorem}
\textbf{Chebyshev's Inequality.} For any random variable $X$ and $t > 0$, $\Pr\left[|X| > t\right] \leq \E[X^2]/{t^2}$.

\subsection{Critical Index}		\label{sec:criticalindex}

Consider a vector of coefficients ${\bf c} \in \mathbbm{R}^{[M] \times [q]}$ such that ${\bf c} = ({\bf c}_i)_{i=1}^M$ and ${\bf c}_i \in \mathbbm{R}^{[q]}$ for each $i \in [M]$. Let $\sigma : [M] \to [M]$ be an ordering such that $\|{\bf c}_{\sigma(1)}\|_2 \geq  \|{\bf c}_{\sigma(2)}\|_2 \geq \ldots \geq \|{\bf c}_{\sigma(M)}\|_2$. For a given $\tau \in (0,1)$, the $\tau$-critical index $i_\tau({\bf c})$ to be the minimum index $i \in [M]$ such that $\|{\bf c}_{\sigma(i)}\|^2 \leq \tau \sum_{i' \geq i}\|{\bf c}_{\sigma(i')}\|^2$. If $i_\tau({\bf c}) = 1$ then ${\bf c}$ is said to be $\tau$-regular. Define $C_\tau({\bf c}) := \left\{\sigma(i) : i < i_\tau({\bf c}) \right\}$, and for any integer $K \in \mathbbm{N}$, $C^{\leq K}_\tau ({\bf c}) = \left\{\sigma(i): 1 \leq i \leq K, i < i_\tau({\bf c}) \right\}$ to denote the set of the first $K$ indices in $C_\tau({\bf c})$ . The following proposition summarizes well known properties of critical indices:

\begin{proposition}[\cite{Serv06}]					\label{prop:crit-ix}
	For the above setting the following condition hold:
	\begin{itemize}
		\item For any $1 \leq i_1 < i_2 \leq i_\tau({\bf c})$, we have $\|{\bf c}_{\sigma(i_2)}\|^2 \leq \frac1\tau(1-\tau)^{i_1 - i_2}\|{\bf c}_{\sigma(i_1)}\|^2$.
		\item The vector $\left({\bf c}_i\right)_{i \in [M]\setminus C_\tau({\bf c})}$ is $\tau$-regular.
	\end{itemize}
\end{proposition}

\subsection{Smooth Label Cover}

Our reduction is from the following hypergraph variant of Smooth Label Cover.

\begin{definition}[\slc]
	A $\slc$ instance  $\mc{L}((V_{\mc{L}},E_{\mc{L}}),\\M,m,\{\pi_{e,v}\}_{e \in E_{\mc{L}},v \in e})$ consists of a regular, $2k$-uniform, connected hypergraph with a vertex set $V_{\mc{L}}$, a hyperedge set $E_{\mc{L}}$ and a set of projections $\{\pi_{e,v} : [M] \mapsto [m]\}_{e \in E_{\mc{L}},v \in e}$.
	
	A labeling $\sigma:V_{\mc{L}} \mapsto [M]$ is said to {\em strongly satisfy} a hyperedge $e$ if for every $v,v' \in e$, we have $\pi_{e,v}(\sigma(v)) = \pi_{e,v'}(\sigma(v'))$. It is said to {\em weakly satisfy} a hyperedge $e$ if there exists a distinct pair of vertices $v,v' \in e$ such that $\pi_{e,v}(\sigma(v)) = \pi_{e,v'}(\sigma(v'))$.
\end{definition}

The following theorem gives the hardness of \slc.

\begin{theorem}[\cite{FGRW12}]				\label{thm:slc-hardness}
	There exists an absolute constant $\gamma_0 > 0$ such that for all integer parameters $z$ and $J$, it NP-Hard to distinguish whether an instance $\mc{L}$ of $\slc$ with $M = 7^{(J+1)z}$ and $N = 2^z7^{Jz}$, satisfies, 
	\begin{itemize}
		\item (YES): There exists a labeling $\sigma:V \mapsto [M]$ which strongly satisfies all the hyperedges.
		\item (NO): There is no labeling $\sigma:V \mapsto [M]$ which weakly satisfies more than $2k^22^{-\gamma_0 z}$-fraction of hyperedges.
	\end{itemize}
	Additionally, $\mc{L}$ satisfies the following properties:
	\begin{itemize}
		\item (Smoothness) For every vertex $v \in V$ and for a randomly sampled hyperedge incident on $v$ 
		\[
		\Pr_{e \sim v} \left[ \pi_{e,v}(i) = \pi_{e,v}(j)\right] \leq \frac{1}{J}
		\]
		for any fixed pair of distinct labels $i,j \in [M]$.
		\item For any hyperedge $e \in E$ and vertex $v \in e$, and any label on the smaller side $\alpha \in [m]$,  $|\pi^{-1}_{e,v}(\alpha)| \leq d$, where $d = 4^z$.
	\end{itemize}
\end{theorem}

\section{Hardness Reduction}				\label{sec:main-red}

In this section we prove the following theorem.
\begin{theorem}			\label{thm:main-red}
	For any constants $\nu, \zeta \in (0,1/2)$ and $\ell \in \mathbbm{N}$, there exists a choice of $z,J$ and $k$ in Theorem \ref{thm:slc-hardness} and a polynomial time reduction from the corresponding \slc\ instance $\mc{L}$ to a distribution $\mc{D}$ supported on point label pairs in $\{0,1\}^n \times \{0,1\}$ such that 
	\begin{itemize}
		\item {\bf Completeness:} If $\mc{L}$ is a YES instance, then there exists two {\sf ORs} $\mc{C}_1$ and $\mc{C}_2$ such that 
				\[
		\Pr_{(\mb{x}, a) \gets \mc{D}} \left[\left(\mc{C}_1(\mb{x}) \wedge \mc{C}_2(\mb{x})\right) = a\right] = 1, \ \ \ \ \tn{and} \ \ \ \ 
		\Pr_{(\mb{x}, a) \gets \mc{D}} \left[\mc{C}_1(\mb{x}) = a\right] \geq 1 - \zeta.
				\]
		\item {\bf Soundness:} If $\mc{L}$ is a NO instance, then for any function $f:\{0,1\}^n \mapsto \{0,1\}$ of $\ell$-halfspaces 				
					\[\Pr_{(\mb{x}, a) \gets \mc{D}} \left[f(\mb{x}) = a\right] \leq \frac12 + \nu.\]
	\end{itemize}
\end{theorem}
Since the negation of an $s$-term DNF is an $s$-clause CNF, the above theorem along with Theorem \ref{thm:slc-hardness} implies Theorem \ref{thm:mainyesno}. In the remainder of the section we provide a detailed construction of the distribution towards proving the above theorem.

The hardness reduction is from an instance $\mc{L}$ of $\slc$ as given in Theorem \ref{thm:slc-hardness} instantiated with  
\begin{equation}				
	J = \frac{10^3\cdot\ell^2\cdot\log^2(dk)\cdot d^{20}}{\nu(\zeta(1-\zeta))^2} \label{eqn:paramJ} 
\end{equation}
where $\gamma_0$ and $d$ are as given in Theorem \ref{thm:slc-hardness}, and we shall the fix the parameter $z$ later. Additionally, we also define the parameters 
\begin{equation}				\label{eqn:param}
	Q = 16dk, \tn{\ \ \ \ } \tau = (10 k\log Q)^{-2}, \tn{\ \ \ \ } K = \frac{20}{\tau}\log\frac{Q}{\tau}, \tn{\ \ \ \ } k = 10/(\zeta(1-\zeta))^2, \tn{\ \ \ \ } t = k/4.  
\end{equation}

For every hyperedge $e$, we define an arbitrary partition of its $2k$ vertices into two subsets of size $k$ each, given by $e = e_X \cup e_Y$ (this notation shall become clear below). 

\medskip
{\it Coordinate Set}:
For every vertex $v \in V_{\mc{L}}$ of the $\slc$ instance and every label $i \in [M]$ we introduce two vectors of $Q$ boolean valued variables (coordinates) each:   $\mb{X}_{v,i}  := \left(X_{v,i,1}, \dots, X_{v,i,Q}\right)$ and $\mb{Y}_{v,i} := \left(Y_{v,i,1}, \dots, Y_{v,i,Q}\right)$. Let $\mb{X}_v := \bigoplus_{i=1}^M \mb{X}_{v,i}$, and $\mb{Y}_v := \bigoplus_{i=1}^M \mb{Y}_{v,i}$. Let ${\bf X}$ and ${\bf Y}$ denote $\{{\bf X_v}\}_{v\in V_{\mc{L}}}$ and $\{{\bf Y}_v\}_{v\in V_{\mc{L}}}$.   In particular, the points of the instance lie in the $(2qM|V_{\mc{L}}|)$-dimensional boolean space.

\medskip
{\it Point-label distribution}: The point-label distribution $\mc{D}_{\sf global}$ is given in Figure \ref{fig:testdist}, with $t$ being a parameter to be decided later. 

\begin{figure}[h!]
	\begin{mdframed}
		\begin{itemize}
			\item [1.] Sample a random hyperedge $e \sim E$. Let $\pi_v := \pi_{e,v}$ for all $v \in e$. Recall the predefined partition of the vertices of $e$ as $e = e_X \cup e_Y$ where $|e_X| = |e_Y| = k$\; 

			\item [2.]  Set ${\bf X}_v = {\bf Y}_v = {\bf 0}$ for all $v \not\in e$\;

			\item [3.] Sample $a \in \{0,1\}$ u.a.r. as the label of the point sampled below\;
			
			\item [4.] For every label $j \in [m]$, do the following:
				\begin{itemize}
					\item[4.1.] Independently sample $S_j \overset{\rm u.a.r}{\sim} {e_X \choose t}$ and $S'_j \overset{\rm u.a.r}{\sim} {e_Y \choose t}$\;
					
					\item[4.2.] Sample a bit $b_j$ to be $0$ w.p. $\zeta$ and $1$ w.p. $1-\zeta$ . If $b_j = 0$, for every $v \in e_X\setminus S_j$ and $i \in \pi^{-1}_{v}(j)$, sample  ${\bf X}_{v,i} \overset{\rm u.a.r}{\sim}  \{0,1\}^Q$. If $b_j=1$, for every $v \in e_Y \setminus S'_j$ and $i \in \pi^{-1}_{v}(j)$, sample ${\bf Y}_{v,i} \overset{\rm u.a.r}{\sim}  \{0,1\}^Q$\;

		\item[4.3] If $a=1$ do Step 4.4 otherwise do Step 4.5\;

		\item[4.4]	{\bf Sampling a $1$-Point} Do the following:
			\begin{enumerate*}
				\item Independently sample $u_{X,j} \sim S_j$ and $u_{Y,j} \sim S'_j$ u.a.r.
				\item For every $i \in \pi^{-1}_{u_{X,j}}(j)$, set ${\bf X}_{u_{X,j},i} \overset{\rm u.a.r}{\sim} \{{\bf e}_1,{\bf e}_2,\ldots,{\bf e}_Q\}$\; 
				\item For every $i \in \pi^{-1}_{u_{Y,j}}(j)$, set ${\bf Y}_{u_{Y,j},i} \overset{\rm u.a.r}{\sim} \{{\bf e}_1,{\bf e}_2,\ldots,{\bf e}_Q\}$\; 
			\end{enumerate*}
			
		\item[4.5]	{\bf Sampling a $0$-Point}  With probability $1 - 1/(\zeta(1-\zeta)t)$, for every $u \in S_j$ and $v \in S'_j$, set ${\bf X}_{u,i}$ and ${\bf Y}_{v,i'}$ to ${\bf 0}$ for each $i \in \pi^{-1}_{u}(j)$ and $i' \in \pi^{-1}_{u}(j)$ . With the rest of the probability $1/(\zeta(1-\zeta)t)$, do the following:
			\begin{enumerate*}
				\item If $b_j = 0$, choose $T_j\subseteq S_j$ by independently sampling vertices of $S_j$ w.p. $(1-\zeta)$. For every $v \in T_j$ and $i \in \pi^{-1}_{v}(j)$, sample  ${\bf X}_{v,i} \overset{\rm u.a.r}{\sim} \{{\bf e}_1,{\bf e}_2,\ldots,{\bf e}_Q\}$\; 
				\item If $b_j = 1$, choose $T'_j\subseteq S'_j$ by independently sampling vertices of $S'_j$ w.p. $\zeta$. For every $v \in T'_j$ and $i \in \pi^{-1}_{v}(j)$, sample ${\bf Y}_{v,i} \overset{\rm u.a.r}{\sim} \{{\bf e}_1,{\bf e}_2,\ldots,{\bf e}_Q\}$\;
			\end{enumerate*}
				\end{itemize}
			\item[5.] For those ${\bf X}_{v, i}$, ${\bf Y}_{v, i}$ not assigned values yet, set them to ${\bf 0}$. Output the point $({\bf X}, {\bf Y})$ with label $a$.
		\end{itemize}
	\end{mdframed}
	\caption{Distribution $\mc{D}_{\sf global}$}
	\label{fig:testdist}
\end{figure}
		The restriction of $\mc{D}_{\sf global}$ to a hyperedge $e$ is denoted by $\mc{D}^e$.  The restrictions of $\mc{D}^e$ to the $0$ and $1$ points are given by $\mc{D}^e_0$ and $\mc{D}^e_1$ respectively. Fix a hyperedge $e$,and some $v \in e_X$. 
		
		Consider the distributions $\mc{D}^e_0$ and $\mc{D}^e_1$ conditioned on a fixation of $\{S_j, S'_j\}_{j\in [m]}$. Under both distributions, if $b_j = 0$ then for each $v \in S_j$ independently each of $\{{\bf X}_{v, i} : i \in \pi^{-1}_{e,v}(j)\}$ is independently sampled u.a.r. from $\{0,1\}^Q$, and if $b_j = 0$ they are all  set to ${\bf 0}$. If $v \in S_j$, with probability exactly $1/t$  each of $\{{\bf X}_{v, i} : i \in \pi^{-1}_{e,v}(j)\}$ is independently sampled u.a.r from $\{{\bf e}_1,{\bf e}_2,\ldots,{\bf e}_Q\}$, and otherwise set to ${\bf 0}$. The same analogously holds true for  $\{{\bf Y}_{v, i} : i \in \pi^{-1}_{e,v}(j)\}$ for $v \in e_Y$, with $\ol{b_j}$ replacing $b_j$, $S'_j$ replacing $S_j$. Based on this we have the following observation. 
Let ${\bf X}^{(j)}_e := \cup_{v \in e_X}\{{\bf X}_{v, i} : i \in \pi^{-1}_{e,v}(j)\}$ and ${\bf Y}^{(j)}_e := \cup_{v \in e_Y}\{{\bf Y}_{v, i} : i \in \pi^{-1}_{e,v}(j)\}$ be the set of variables for $e$ corresponding to $j \in [m]$.

		\begin{observation} \label{obs:inducedist}
			Conditioned on any fixation of $\{S_j, S'_j\}_{j\in [m]}$ the following holds. 
			\begin{itemize}
				\item Distributions $\mc{D}^e_0$ and $\mc{D}^e_1$ induce the same distribution on any variable ${\bf X}_{v, i}$ or ${\bf Y}_{v, i}$. 
				\item Under both $\mc{D}^e_0$ and $\mc{D}^e_1$, the variables $({\bf X}^{(j)}_e, {\bf Y}^{(j)}_e, b_j)$, $j = 1,\dots, m$ are sampled independently (over the different $j \in [m]$). 
				\item  For each $j \in [m]$, after additionally fixing $b_j$, the variables 
			$$\{{\bf X}_{v, i} :  v \in e_X, v \not\in S_j, i \in \pi^{-1}_{e,v}(j)\} \cup \{{\bf Y}_{v, i} :  v\in e_Y, v\not\in S'_j, i \in \pi^{-1}_{e,v}(j)\}$$ have the same distribution in  $\mc{D}^e_0,$
			and $\mc{D}^e_1$, and are independent of the variables corresponding to $v \in S_j$ and $v \in S_j'$.
			\end{itemize}
		\end{observation}
We denote a halfspace over the variables ${\bf X}$ and ${\bf Y}$ as ${\rm pos}(h({\bf X}, {\bf Y}))$ where $h({\bf X}, {\bf Y}) = \Langle {\bf c}_X, {\bf X}\Rangle + \Langle {\bf c}_Y, {\bf Y}\Rangle + \theta$, where ${\bf c}_X, {\bf c}_Y \in \mathbbm{R}^{V_{\mc{L}}\times [M] \times [q]}$ and $\theta$ is a constant. For any vertex $v \in V_{\mc{L}}$, $i \in [M]$ and $q \in [Q]$ we denote by (i) ${\bf c}_{X,v,i} := (c_{X,v,i,q})_{q=1}^Q$ the vector of coefficients corresponding to ${\bf X}_{v,i}$, and (ii) ${\bf c}_{X,v}$ the vector of coefficients corresponding to ${\bf X}_v$. Similarly for  ${\bf c}_{Y,v,i}$ and ${\bf c}_{Y,v}$.

\section{Completeness}				\label{sec:completeness}
	Suppose $\mc{L}$ is a YES instance. Then there exists labeling $\sigma: V \mapsto [M]$ such that for every hyperedge $e$ we have $j_e \in [m]$ s.t. $\pi_{e,v}(\sigma(v)) = j_e$ for every $v \in e$. Consider the following two {\sf OR} formulas: 
	\[
		\mc{C}_1 = \left(\bigvee_{v \in V}\bigvee_{q \in [Q]} X_{v,\sigma(v),q} \right), \tn{ \ \ and\ \ } 	\mc{C}_2 = \left(\bigvee_{v \in V}\bigvee_{q \in [Q]} Y_{v,\sigma(v),q}\right)
	\] 
	
	Note that $\mc{C}_1$ depends only on ${\bf X}$ and $\mc{C}_2$ only on ${\bf Y}$. Fix a hyperedge $e = e_X \cup e_Y$. 
	
	Suppose $a = 1$. By construction of the distribution, there is $u = u_{X, j_e}$ and $v =  u_{Y, j_e}$ such that for some $q, q' \in [Q]$ we have $X_{u, \sigma(u), q} = 1$ and $Y_{v,\sigma(v),q'} = 1$. This implies that $\mc{C}_1 = \mc{C}_2 = 1$, and $\mc{C}_1\wedge \mc{C}_2$ evaluates to $1$ when $a = 1$.

	On the other hand, suppose  $a = 0$. Note that all the variables corresponding to vertices not in $e$ are set to zero.
	If $b_{j_e} = 0$, then ${\bf Y}_{v,\sigma(v)} = {\bf 0}$ for every $v \in e_Y$ implying that $\mc{C}_2 = 0$. Otherwise if $b_{j_e}=1$ then ${\bf X}_{v,\sigma(v)} = {\bf 0}$ for every $v \in e_X$ so that $\mc{C}_1 = 0$. Therefore, we have $\mc{C}_1\wedge\mc{C}_2 = 0$ whenever $a=0$. Furthermore, since $b_{j_e} = 1$ w.p. $1 - \zeta$, $\mc{C}_1 = 0$ w.p. $1 - \zeta$ on points labeled $0$. 

\section{Soundness Analysis}		\label{sec:soundness}		

	Suppose there exists a function $f$ of $\ell$-halfspaces ${\rm pos}(h_1({\bf X}, {\bf Y})), \dots, {\rm pos}(h_\ell({\bf X}, {\bf Y}))$ such that: 
	\begin{equation}
		\Pr_{(({\bf X},{\bf Y}), a) \sim \mathcal{D}_{\sf global}}\left[f({\bf X},{\bf Y}) = a\right] \geq \frac12 + \nu. 
	\end{equation}		
		Observe that by averaging, for at least $\nu/2$-fraction of the hyperedges $e$, $f$ is consistent with $\mc{D}_e$ with probability at least $\frac12 + \frac{\nu}{2}$. Call such hyperedges \emph{fine} and for each such hyperedge $e$ we have
		\begin{equation}
			\Pr_{(({\bf X},{\bf Y}), a) \sim \mathcal{D}^e}\left[f({\bf X},{\bf Y}) = a\right] \geq \frac12 + \frac{\nu}{2}. \label{eqn:ffine}
		\end{equation}

{\bf Applying Label Cover Smoothness}. Recall that $\{({\bf c}^{(s)}_X, {\bf c}^{(s)}_Y)\}_{s=1}^\ell$ are the $\ell$ coefficient vectors for the variables $({\bf X}, {\bf Y})$. For a vertex $v$, we have the subsets $C_\tau ({\bf c}^{(s)}_{X,v})$, $C_\tau ({\bf c}^{(s)}_{Y,v})$, $C^{\leq K}_\tau ({\bf c}^{(s)}_{X,v})$, and $C^{\leq K}_\tau ({\bf c}^{(s)}_{Y,v})$, for $s \in [\ell]$ as defined in Section \ref{sec:criticalindex}. Additionally, define the following subset of $[M]$:

\begin{equation}
\displaystyle	L_v\left(\{({\bf c}^{(s)}_X, {\bf c}^{(s)}_Y)\}_{s=1}^\ell\right) := \bigcup_{s=1}^\ell I_v({\bf c}^{(s)}_X, {\bf c}^{(s)}_Y), \label{eqn:Lv}
\end{equation}
where
\begin{align}
I_v({\bf c}_X, {\bf c}_Y) =  C^{\leq K}_\tau ({\bf c}_{X,v}) \bigcup C^{\leq K}_\tau ({\bf c}_{Y,v}) & \bigcup \left\{  i \in [M]\setminus C_\tau ({\bf c}_{X,v}): \|{\bf c}_{X,v,i}\|^2 > \frac{1}{d^8}\sum_{i' \notin C_\tau ({\bf c}_{X,v})}\left\|{\bf c}_{X,v,i'}\right\|^2_2\right\} \nonumber \\
\bigcup & \left\{i \in [M]\setminus C_\tau ({\bf c}_{Y,v}) : \|{\bf c}_{Y,v,i}\|^2 > \frac{1}{d^8}\sum_{i' \notin C_\tau ({\bf c}_{Y,v})}\left\|{\bf c}_{Y,v,i'}\right\|^2_2 \right\}. \label{eqn:Iv}
\end{align}

The following is a consequence of the smoothness property of $\slc$ and our setting of the parameters.

\begin{lemma}			\label{lem:nice-hyper}
	Given  $\{({\bf c}^{(s)}_X, {\bf c}^{(s)}_Y)\}_{s=1}^\ell$, at least $(1-\nu/4)$ fraction of the hyperedges $e \in E_{\mc{L}}$ are \emph{nice} w.r.t.  $\{({\bf c}^{(s)}_X, {\bf c}^{(s)}_Y)\}_{s=1}^\ell$  i.e., the following condition holds for every vertex $v \in e$: for any pair of labels $i_1,i_2 \in L_v\left( \{({\bf c}^{(s)}_X, {\bf c}^{(s)}_Y)\}_{s=1}^\ell\right)$, $\pi_{e,v}(i_1) \neq \pi_{e,v}(i_2)$. 
\end{lemma}
\begin{proof} It is easy to see that the size of $I_v({\bf c}_X, {\bf c}_Y)$ for any vertex $v$ is at most $2(K + d^8)$, and thus the size of $L_v = L_v\left(\{({\bf c}^{(s)}_X, {\bf c}^{(s)}_Y)\}_{s=1}^\ell\right)$ is at most $2\ell(K + d^8)$. By the regularity of $G_{\mc{L}} = (V_{\mc{L}}, E_{\mc{L}})$, the smoothness of $\mc{L}$ and union bound, we have,
	\begin{eqnarray}
	& &	\Pr_{e \sim E_{\mc{L}}}\left[\exists v \in e ,\exists i_1,i_2 \in L_v : \pi_{e,v}(i_1) = \pi_{e,v}(i_2)\right] \nonumber \\
	& \leq & 2k\cdot \E_{v \in V_{\mc{L}}}\left[\Pr_{e \sim v}\left[\exists i_1,i_2 \in L_v : \pi_{e,v}(i_1) = \pi_{e,v}(i_2)\right]\right] \nonumber \\
	& \leq &  2k\cdot \E_{v \in V_{\mc{L}}} \left[|L_v|^2/J\right] \nonumber \\
	& \leq & \left(8k\ell^2\cdot(K + d^8)^2\right)/J \leq \nu/4,
	\end{eqnarray}
	where the last inequality is obtained by our setting of $J$ in \eqref{eqn:paramJ}.
\end{proof}
For convenience we shall abuse notation to say that $e$ is nice w.r.t. $({\bf c}_X, {\bf c}_Y)$ if for each $v \in E$, for any pair of labels $i_1,i_2 \in I_v({\bf c}_X, {\bf c}_Y)$, $\pi_{e,v}(i_1) \neq \pi_{e,v}(i_2)$. We have the following observation.

\begin{observation}\label{obs:singlenice}
	If $e$ is nice w.r.t. $\{({\bf c}^{(s)}_X, {\bf c}^{(s)}_Y)\}_{s=1}^\ell$ then it is nice w.r.t. each $({\bf c}^{(s)}_X, {\bf c}^{(s)}_Y)$, $s \in [\ell]$. 
\end{observation}

{\bf The Main Structural Lemma}. Applying Lemma \ref{lem:nice-hyper} to the coefficient vectors $({\bf c}^{(s)}_X, {\bf c}^{(s)}_Y)$ for $h_s$ for $s \in [\ell]$ we obtain that there are at least $(\nu/2 - \nu/4) = (\nu/4)$-fraction of hyperedges which are \emph{fine} and \emph{nice} i.e., they satisfy \eqref{eqn:ffine} as well as the niceness condition of Lemma \ref{lem:nice-hyper} w.r.t. $\{({\bf c}^{(s)}_X, {\bf c}^{(s)}_Y)\}_{s=1}^\ell$. Let $E^*$ represent the set such hyperedges. 
We now state the main structural lemma the proof of which is provided later in this section.
	\begin{lemma}				\label{lem:dec-struct}
For each $e \in E^*$ there exist distinct vertices $u, v \in e$ such that at least one of the following is satisfied,	
		\begin{itemize}
			\item[\tn{I.}] There exist $r, p \in [\ell]$ ($r$ may equal $p$) s.t. 
				$$\pi_{e,u}\left(C^{\leq K}_{\tau}({\bf c}^{(r)}_{X,u})\cup C^{\leq K}_{\tau}({\bf c}^{(r)}_{Y,u})\right) \cap \pi_{e,v}\left(C^{\leq K}_{\tau}({\bf c}^{(p)}_{X,v})\cup C^{\leq K}_{\tau}({\bf c}^{(p)}_{Y,v})\right) \neq \emptyset.$$

			\item[\tn{II.}] For some $r \in [\ell]$, there exists $j \in\pi_{e,u}\left(C^{\leq K}_{\tau}({\bf c}^{(r)}_{X,u})\cup C^{\leq K}_{\tau}({\bf c}^{(r)}_{Y,u})\right)$ such that 
at least one of \eqref{eqn:dec-struct-1} or \eqref{eqn:dec-struct-2} given below is satisfied 
				\begin{eqnarray}	
					\displaystyle j \notin \pi_{e,v}\left(C^{\leq K}_\tau({\bf c}^{(r)}_{X,v})\right), &\tn{\ \ and \ \ }& \sum_{i \in \pi^{-1}_{e,v}(j) \setminus C_\tau({\bf c}^{(r)}_{X,v})}\|{\bf c}^{(r)}_{X, v,i}\|^2\, >\, \tau^4 \sum_{i \in [M]\setminus C_\tau({\bf c}^{(r)}_{X,v})}\|{\bf c}^{(r)}_{X, v,i}\|^2 \label{eqn:dec-struct-1}\\
					\displaystyle j \notin \pi_{e,v}\left(C^{\leq K}_\tau({\bf c}^{(r)}_{Y,v})\right), &\tn{\ \ and \ \ }& \sum_{i \in \pi^{-1}_{e,v}(j) \setminus C_\tau({\bf c}^{(r)}_{Y,v})}\|{\bf c}^{(r)}_{Y, v,i}\|^2 \, > \, \tau^4 \sum_{i \in [M]\setminus C_\tau({\bf c}^{(r)}_{Y,v})}\|{\bf c}^{(r)}_{Y, v,i}\|^2 \label{eqn:dec-struct-2}
				\end{eqnarray}
		\end{itemize}  
	\end{lemma}

\medskip
\noindent
{\bf Labeling the Vertices of $\mc{L}$ using Lemma \ref{lem:dec-struct}.}			
Consider the following randomized labeling for each $v \in V_{\mc{L}}$: with probability $1/2$ each do Step 2 or 3.

\begin{enumerate}
	\item Choose $s \in  [\ell]$ u.a.r.
	\item Assign $\sigma(v) \overset{\rm unif}{\sim} \left(C^{\leq K}_{\tau}({\bf c}^{(s)}_{X,v}) \cup C^{\leq K}_{\tau}({\bf c}^{(s)}_{Y,v})\right)$.
	\item W.p. $1/2$ each do (a) or (b):
	\begin{enumerate}
		\item Assign $\sigma(v)$ a label $i \notin C_\tau({\bf c}^{(s)}_{X,v})$ with probability 
			$$\|{\bf c}^{(s)}_{X,v,i}\|^2/\left(\sum_{i \in [M]\setminus C_\tau({\bf c}^{(s)}_{X,v})}\|{\bf c}^{(s)}_{X, v,i}\|^2\right).$$
		\item Assign $\sigma(v)$ a label $i \notin C_\tau({\bf c}^{(s)}_{Y,v})$ with probability 
			$$\|{\bf c}^{(s)}_{Y,v,i}\|^2/\left(\sum_{i \in [M]\setminus C_\tau({\bf c}^{(s)}_{Y,v})}\|{\bf c}^{(s)}_{Y, v,i}\|^2\right).$$
	\end{enumerate}
\end{enumerate}
 
We now analyze the probability of the above labeling \emph{weakly} satisfying a fixed hyperedge $e \in E^*$. Suppose that $e$ satisfies Case I of Lemma \ref{lem:dec-struct} for $u, v \in e$. Then, with probability $\geq 1/(4\ell^2)$ both $u$ and $v$ are labeled using Step 1 of the  labeling procedure with $u$ choosing $r$ and $v$ choosing $p$. Therefore, in this case, $\pi_{e, u}(\sigma(u)) = \pi_{e, v}(\sigma(v))$ with probability at least
\begin{equation}
	\frac{1}{4\ell^2}\cdot\frac{1}{\left|C^{\leq K}_{\tau}({\bf c}^{(r)}_{X,u}) \cup C^{\leq K}_{\tau}({\bf c}^{(r)}_{Y,u})\right|} \cdot	\frac{1}{\left|C^{\leq K}_{\tau}({\bf c}^{(p)}_{X,v}) \cup C^{\leq K}_{\tau}({\bf c}^{(p)}_{Y,v})\right|} \geq \frac{1}{16 K^2\ell^2}.
\end{equation}
On the other hand, suppose $e$ satisfies Case II of Lemma \ref{lem:dec-struct} such that \eqref{eqn:dec-struct-1} holds, for $u, v \in e$. Then, with probability $1/(8\ell^2)$ $u$ is labeled according to Step 1 using $s=r$ and $v$ according to Step 2.a of the labeling procedure using $s=r$. Thus, in this case $\pi_{e, u}(\sigma(u)) = \pi_{e, v}(\sigma(v))$ with probability at least
\begin{equation}
	\frac{1}{8\ell^2}\cdot\frac{1}{\left|C^{\leq K}_{\tau}({\bf c}^{(r)}_{X,u}) \cup C^{\leq K}_{\tau}({\bf c}^{(r)}_{Y,u})\right|}\cdot \tau^4 \geq \frac{\tau^4}{16 K \ell^2},
\end{equation}
which also analogously holds when $e$ satisfies Case II of Lemma \ref{lem:dec-struct} such that \eqref{eqn:dec-struct-2} is true.

Combining the above with the lower bound on the size of $E^*$, we obtain that the expected fraction of hyperedges satisfied by the labeling is at least,
\begin{equation}
	\frac{1}{16\ell^2}\cdot\frac{\nu}{4}\cdot\min\left\{\frac{1}{K^2}, \frac{\tau^4}{K}\right\}
\end{equation}
Plugging in the parameters from \eqref{eqn:param} and $d = 4^z$ (as given in Theorem \ref{thm:slc-hardness}), the RHS of the above equation is at least $\nu(10k\log(16dk))^{-10} \geq \nu(10kz\log (64k))^{-10}$ (from the value of $d$ in Theorem \ref{thm:slc-hardness}). On the other hand, from the guarantee of the NO case of Theorem \ref{thm:slc-hardness}, we know that any labeling can weakly satisfy at most $2k^22^{-\gamma_0 z}$ fraction of hyperedges in $\mc{L}$. By choosing $z$ to be a large enough constant, the expected number of weakly satisfied hyperedges exceeds the soundness parameter, which completes the proof of Theorem \ref{thm:main-red}.

\subsection{Proof of Lemma \ref{lem:dec-struct}}
Fix an edge $e \in E^*$. Assume for a contradiction that $e$ does not satisfy Lemma \ref{lem:dec-struct}.
For convenience, define 
\begin{equation}
	B_{s, v} := C^{\leq K}_{\tau}({\bf c}^{(s)}_{X,v})\cup C^{\leq K}_{\tau}({\bf c}^{(s)}_{Y,v})
\end{equation}
for $s \in [\ell]$ and $v \in e$.
Let us say that two triples $(r, u, i_1)$ and $(p, v, i_2)$ from $[\ell]\times e \times [M]$ are distinct if they differ in at least one coordinate.  From the niceness of $e$ w.r.t. $\{({\bf c}^{(s)}_X, {\bf c}^{(s)}_Y)\}_{s=1}^\ell$ given by Lemma \ref{lem:nice-hyper} and the negation of Condition I of Lemma \ref{lem:dec-struct} we have the following observation. 

\begin{observation} \label{obs:triples}
	For any two distinct triples $(r, u, i_1)$ and $(p, v, i_2)$ s.t. $i_1 \in B_{r, u}$ and $i_2 \in B_{p, v}$, $\pi_{e, u}(i_1) \neq \pi_{e, v}(i_2)$. 
\end{observation}

Combining the above with Observation \ref{obs:inducedist} we obtain that after any fixation of $\{S_j, S_j'\}_{j=1}^m$ the variables 
\begin{equation}
	\{{\bf X}_{v, i}\ :\ v \in e_X,s \in [\ell], i \in C^{\leq K}_{\tau}({\bf c}^{(s)}_{X,v})\} \cup \{{\bf Y}_{v, i}\ :\ v \in e_Y, s \in [\ell], i \in C^{\leq K}_{\tau}({\bf c}^{(s)}_{Y,v})\}, \label{eqn:disjvar}
\end{equation}
are independently distributed under both $\mc{D}^e_0$ and $\mc{D}^e_1$. Further, since their marginal distributions match under  $\mc{D}^e_0$ and $\mc{D}^e_1$ (again by  Observation \ref{obs:inducedist}), their joint distributions also match. Using this we describe in Figure \ref{fig:conddist} a distribution $\wh{\mc{D}}_e$ on $(({\bf X}^0, {\bf Y}^0), ({\bf X}^1, {\bf Y}^1))$ with  marginal distributions $\mc{D}^e_0$ and $\mc{D}^e_1$ respectively.
\hspace{1cm}
\begin{figure}[h!]
	\begin{mdframed}
\begin{itemize}
	\item[1.] Sample $\{S_j, S_j'\}_{j=1}^m$ as according to $\mc{D}^e$.
	\item[2.] Conditioned on $\{S_j, S_j'\}_{j=1}^m$, sample the variables in \eqref{eqn:disjvar} by sampling $b_j$ for $j \in \cup_{v\in e}\cup_{s\in[\ell]}\pi_{e,v}(B_{s,v})$.
	\item[3.] Conditioned on the fixing till now sample all the rest of the $b_j$ and then the rest of the variables in
		$$\{{\bf X}_{v, i} :  v \in e_X, v \not\in S_j, i \in \pi^{-1}_{e,v}(j)\} \cup \{{\bf Y}_{v, i} :  v\in e_Y, v\not\in S'_j, i \in \pi^{-1}_{e,v}(j)\},$$
		according to the distribution common to $\mc{D}^e_0$ and $\mc{D}^e_1$ (see  Observation \ref{obs:inducedist}). 
	\item[4.] Call the above fixings $\Gamma_e$. Conditioned on $\Gamma_e$ sample $({\bf X}^0, {\bf Y}^0) \gets \mc{D}^e_0$ and $({\bf X}^1, {\bf Y}^1) \gets \mc{D}^e_1$.
	\item[5.] Output $(({\bf X}^0, {\bf Y}^0), ({\bf X}^1, {\bf Y}^1))$.
\end{itemize}
\end{mdframed}
	\caption{Distribution $\wh{\mc{D}}_e$}
	\label{fig:conddist}
\end{figure}

From \eqref{eqn:ffine} we obtain that,
\begin{equation}
	\E_{(({\bf X}^0, {\bf Y}^0), ({\bf X}^1, {\bf Y}^1))\gets \wh{\mc{D}}_e}\left[\left|f\left({\bf X}^1, {\bf Y}^1\right) - f\left({\bf X}^0, {\bf Y}^0\right)\right|\right] \geq \frac{\nu}{2}.
\end{equation}
By averaging, there is one halfspace ${\rm pos}(h)(\cdot)$  on which $f$ depends, satisfying:
\begin{equation}
	\E_{(({\bf X}^0, {\bf Y}^0), ({\bf X}^1, {\bf Y}^1))\gets \wh{\mc{D}}_e}\left[\left|{\rm pos}\left(h\left({\bf X}^1, {\bf Y}^1\right)\right) - {\rm pos}\left(h\left({\bf X}^0, {\bf Y}^0\right)\right)\right|\right] \geq \frac{\nu}{2\ell}. \label{eqn:nu2ell}
\end{equation}
Let the linear form $h$ be given by $h({\bf X}, {\bf Y}) = \Langle {\bf c}_{X}, {\bf X}\Rangle + \Langle {\bf c}_{Y}, {\bf Y}\Rangle + \theta$.

Applying Lemma \ref{lem:trunc} iteratively for each $v \in e$ on $h$, we obtain a \emph{truncated} linear form $\tilde{h}({\bf X}, {\bf Y}) = \Langle \tilde{{\bf c}}_{X}, {\bf X}\Rangle + \Langle \tilde{{\bf c}}_{Y}, {\bf Y}\Rangle + \theta$ satisfying 
\begin{itemize}
	\item[(A)] $\tilde{\bf c}_{X, u} = {\bf c}_{X, u}$ and $\tilde{\bf c}_{Y, u} = {\bf c}_{Y, u}$ for all $u \in V_{\mc{L}}, u \not\in e$. 
	\item[(B)] For $v \in e_X$: $\tilde{\bf c}_{X, v, i} = {\bf 0}$ for all $i \in C_{\tau}({\bf c}_{X,v})\setminus C_{\tau}^{\leq K}({\bf c}_{X,v})$ and $\tilde{\bf c}_{X, v, i} = {\bf c}_{X, v, i}$ otherwise,
	\item[(C)] For $v \in e_Y$: $\tilde{\bf c}_{Y, v, i} = {\bf 0}$ for all $i \in C_{\tau}({\bf c}_{Y,v})\setminus C_{\tau}^{\leq K}({\bf c}_{Y,v})$ and $\tilde{\bf c}_{Y, v, i} = {\bf c}_{Y, v, i}$ otherwise,
\end{itemize}
and by union bound,
\begin{equation}
	\E_{(({\bf X},{\bf Y}), a) \sim \mathcal{D}_e}\left[\left|{\rm pos}\left(h\left({\bf X},{\bf Y}\right)\right) - {\rm pos}\left(\tilde{h}\left({\bf X},{\bf Y}\right)\right)\right|\right] \leq 2k\cdot\frac{\tau^{1/4}}{4k} = \frac{\tau^{1/4}}{2}. \label{eqn:applytrucn}
\end{equation}
Using  the above  we obtain,
\begin{eqnarray}
	& & \E_{(({\bf X}^0, {\bf Y}^0), ({\bf X}^1, {\bf Y}^1))\gets \wh{\mc{D}}_e}\left[\left|{\rm pos}\left(\tilde{h}\left({\bf X}^1, {\bf Y}^1\right)\right) - {\rm pos}\left(\tilde{h}\left({\bf X}^0, {\bf Y}^0\right)\right)\right|\right] \nonumber \\ 
	& \geq & \E_{(({\bf X}^0, {\bf Y}^0), ({\bf X}^1, {\bf Y}^1))\gets \wh{\mc{D}}_e}\left[\left|{\rm pos}\left(h\left({\bf X}^1, {\bf Y}^1\right)\right) - {\rm pos}\left(h\left({\bf X}^0, {\bf Y}^0\right)\right)\right|\right] \nonumber \\ & & \ \  - \left(\E_{(({\bf X}^0,{\bf Y}^0)) \sim \mathcal{D}^0_e}\left[\left|{\rm pos}\left(h\left({\bf X}^0,{\bf Y}^0\right)\right) - {\rm pos}\left(\tilde{h}\left({\bf X}^0,{\bf Y}^0\right)\right)\right|\right]  \right. \nonumber \\
	& & \ \ \ \ \left. + \E_{(({\bf X}^1,{\bf Y}^1)) \sim \mathcal{D}^1_e}\left[\left|{\rm pos}\left(h\left({\bf X}^1,{\bf Y}^1\right)\right) - {\rm pos}\left(\tilde{h}\left({\bf X}^1,{\bf Y}^1\right)\right)\right|\right]\right) \nonumber \\
	& = & \frac{\nu}{2\ell} - 2\E_{(({\bf X},{\bf Y}), a) \sim \mathcal{D}_e}\left[\left|{\rm pos}\left(h\left({\bf X},{\bf Y}\right)\right) - {\rm pos}\left(\tilde{h}\left({\bf X},{\bf Y}\right)\right)\right|\right] \label{eqn:diffbd} \\
	& \geq & \frac{\nu}{2\ell} - \tau^{1/4}, \label{eqn:diffbd1}
\end{eqnarray}
where the equality \eqref{eqn:diffbd} is due to $a$ being $0$ or $1$ with equal probability under $\mc{D}^e$, and the final inequality uses \eqref{eqn:nu2ell}.

From the structural properties (A), (B) and (C) of $\tilde{h}$ listed above it is easy to see that $(\tilde{{\bf c}}_{X}, \tilde{{\bf c}}_{Y})$ satisfies \eqref{eqn:lemstructsetting}, and
\begin{align}
	C_\tau(\tilde{{\bf c}}_{X,v}) = C^{\leq K}_\tau(\tilde{{\bf c}}_{X,v}) = C^{\leq K}_\tau({\bf c}_{X,v}),& \tn{ and \ \ } \sum_{i \notin C_\tau (\tilde{{\bf c}}_{X,v})}\left\|\tilde{{\bf c}}_{X,v,i}\right\|^2_2 = \sum_{i \notin C_\tau ({\bf c}_{X,v})}\left\|{\bf c}_{X,v,i}\right\|^2_2 \tn{\ for\ } v \in e_X, \label{eqn:additional1} \\
	C_\tau(\tilde{{\bf c}}_{Y,v}) = C^{\leq K}_\tau(\tilde{{\bf c}}_{Y,v}) = C^{\leq K}_\tau({\bf c}_{Y,v}),& \tn{ and \ \ } \sum_{i \notin C_\tau (\tilde{{\bf c}}_{Y,v})}\left\|\tilde{{\bf c}}_{Y,v,i}\right\|^2_2 = \sum_{i \notin C_\tau ({\bf c}_{Y,v})}\left\|{\bf c}_{Y,v,i}\right\|^2_2 \tn{\ for\ } v \in e_Y,
	\label{eqn:additional2}
\end{align}
Together with (B) and (C), \eqref{eqn:additional1} and \eqref{eqn:additional2} imply that $I_v({\bf c}_X, {\bf c}_Y) \supseteq I_v(\tilde{\bf c}_X, \tilde{\bf c}_Y)$ and thus $e$ remains nice w.r.t. $(\tilde{{\bf c}}_{X}, \tilde{{\bf c}}_{Y})$. From our assumption that $e$ does not satisfy the conditions from  Lemma \ref{lem:dec-struct}, we claim that $(\tilde{{\bf c}}_{X}, \tilde{{\bf c}}_{Y})$ satisfies both the conditions of Lemma  \ref{lem:main-struct}. In particular, by negating Condition $\tn{I}$ of Lemma \ref{lem:dec-struct}, for any pair of vertices $u,v \in e$ we have 
\begin{align*}
& \pi_{e,u}\left(C^{\leq K}_\tau({\bc}_{X,u}) \cup C^{\leq K}_\tau({\bc}_{{Y},u})\right) \cap \pi_{e,v}\left(C^{\leq K}_\tau({\bc}_{X,v}) \cup C^{\leq K}_\tau({{\bf c}}_{Y,v})\right) = \emptyset \\
\Rightarrow\ &\pi_{e,u}\Big(C_\tau(\tilde{{\bf c}}_{X,u}) \cup C_\tau(\tilde{{\bf c}}_{{Y},u})\Big) \cap \pi_{e,v}\Big(C_\tau(\tilde{{\bf c}}_{X,v}) \cup C_\tau(\tilde{{\bf c}}_{ Y,v})\Big)  = \emptyset
\end{align*}
which gives us Condition 1 of Lemma \ref{lem:main-struct}. Towards establishing the Condition 2 of Lemma \ref{lem:main-struct}, we observe that for $P$ as used in Lemma \ref{lem:main-struct}, any $v \in e_X$ and $j \in P \setminus\pi_{e,v}\left(C^{\leq K}_\tau(\tilde{\bc}_{X,v})\right)$ we have
\begin{align*}
	\sum_{i \in \pi^{-1}_{e,v}(j) \setminus C_\tau(\tilde{\bc}_{X,v})} \|\tilde{\bc}_{X,v,i}\|^2  = & \sum_{i \in \pi^{-1}_{e,v}(j) \setminus C_\tau({\bc}_{X,v})} \|\tilde{\bc}_{X,v,i}\|^2\  + \sum_{i \in \pi^{-1}_{e,v}(j)\cap(C_\tau({\bc}_{X,v})\setminus  C_\tau^{\leq K}({\bc}_{X,v}))} \|\tilde{\bc}_{X,v,i}\|^2 \\ 
	= & \sum_{i \in \pi^{-1}_{e,v}(j) \setminus C_\tau({\bc}_{X,v})} \|{\bc}_{X,v,i}\|^2 \leq \tau^4\sum_{i \notin C_\tau({\bc}_{X,v})} \|{\bc}_{X,v,i}\|^2 = \tau^4\sum_{i \notin C_\tau(\tilde{\bc}_{X,v})} \|\tilde{\bc}_{X,v,i}\|^2
\end{align*}
where the first equality  uses the first part of \eqref{eqn:additional1}, the second and the last equalities follow property (B) of $\tilde{h}$, and the inequality uses our assumption that $e$ does not satisfy Condition $\tn{II}$ from Lemma \ref{lem:dec-struct}. Similar guarantees also hold for any $v \in e_Y$, which together give us Condition 2 of Lemma \ref{lem:main-struct}. Therefore, by applying Lemma \ref{lem:main-struct} we get that
$$\E_{\widehat{\mc{D}}}\left[\left|{\rm pos}\left(\tilde{h}\left({\bf X}^0,{\bf Y}^0\right)\right) - {\rm pos}\left(\tilde{h}\left({\bf X}^1,{\bf Y}^1\right)\right)\right|\right] \leq O(\tau),$$
which is a contradiction to \eqref{eqn:diffbd1} for  small enough setting of $\tau$, which can be achieved through \eqref{eqn:param} by setting $z$ in Theorem \ref{thm:slc-hardness} large enough.

\section{Truncating Long Critical Index Lists}						\label{sec:trun-crit-index}

Fix a vertex $v \in V_{\mc{L}}$. 
Consider a linear form $h$ given by 
$h({\bf X}, {\bf Y}) = \Langle {\bf c}_{X}, {\bf X}\Rangle + \Langle {\bf c}_{Y}, {\bf Y}\Rangle + \theta$.
Let $\tilde{h}$ be the linear form of \emph{truncated} coefficient vectors given by  $\tilde{h}({\bf X}, {\bf Y}) = \Langle \tilde{{\bf c}}_{X}, {\bf X}\Rangle + \Langle \tilde{{\bf c}}_{Y}, {\bf Y}\Rangle + \theta$,
where
\begin{itemize}
	\item $\tilde{\bf c}_{X, u} = {\bf c}_{X, u}$ and $\tilde{\bf c}_{Y, u} = {\bf c}_{Y, u}$ for all $u \in V_{\mc{L}}, u \neq v$. 
	\item For $v \in e_X$: $\tilde{\bf c}_{X, v, i} = {\bf 0}$ for all $i \in C_{\tau}({\bf c}_{X,v})\setminus C_{\tau}^{\leq K}({\bf c}_{X,v})$ and $\tilde{\bf c}_{X, v, i} = {\bf c}_{X, v, i}$ otherwise.
	\item For $v \in e_Y$: $\tilde{\bf c}_{Y, v, i} = {\bf 0}$ for all $i \in C_{\tau}({\bf c}_{Y,v})\setminus C_{\tau}^{\leq K}({\bf c}_{Y,v})$ and $\tilde{\bf c}_{Y, v, i} = {\bf c}_{Y, v, i}$ otherwise.
\end{itemize}
This section proves the following lemma.
\begin{lemma}					\label{lem:trunc}
	Given the above setting, for any hyperedge $e$, such that $e$ is nice w.r.t. $({\bf c}_{X}, {\bf c}_{Y})$ (as given in Lemma \ref{lem:nice-hyper}), and $v \in e$ fixed above, the following holds:
	\begin{equation}
		\E_{(({\bf X},{\bf Y}), a) \sim \mathcal{D}_e}\left[\left|{\rm pos}\left(h\left({\bf X},{\bf Y}\right)\right) - {\rm pos}\left(\tilde{h}\left({\bf X},{\bf Y}\right)\right)\right|\right] \leq \frac{\tau^{1/4}}{4k}. \label{eqn:lemmatrucn}
	\end{equation}
\end{lemma}
\begin{proof} 
	Assume that $v \in e_X$ (we shall handle the $v \in e_Y$ case analogously). Given this, we may further assume that $C_{\tau}({\bf c}_{X,v})\setminus C_{\tau}^{\leq K}({\bf c}_{X,v}) \neq \emptyset$ implying that $\left|C_{\tau}^{\leq K}({\bf c}_{X,v})\right| = K$, otherwise $h = \tilde{h}$ under $\mathcal{D}_e$.
	For ease of notation we  relabel the indices in $[M]$ so that $C_{\tau}^{\leq K}({\bf c}_{X,v}) = [K]$, and denote $\pi = \pi_{e,v}$. First we  bound the difference between $h$ and $\tilde{h}$ as follows:
	\begin{eqnarray}
		\left|h\left({\bf X},{\bf Y}\right) - \tilde{h}\left({\bf X},{\bf Y}\right)\right| & \leq & \sum_{i \in C_{\tau}({\bf c}_{X,v})\setminus [K]}  \Langle {\bf c}_{X, v, i}, {\bf X}_{v,i} \Rangle \nonumber \\
		& \leq & \sum_{i \in C_{\tau}({\bf c}_{X,v})\setminus [K]} \|{\bf c}_{X, v, i}\|_1 \nonumber \\
		\tn{(By Cauchy-Schwartz)}  & \leq & \sqrt{Q}\sum_{i \in C_{\tau}({\bf c}_{X,v})\setminus [K]} \|{\bf c}_{X, v, i}\|_2 \nonumber \\
		\tn{(By Proposition \ref{prop:crit-ix})}		& \leq & \sqrt{Q/\tau} (1 - \tau)^{K/4} \|{\bf c}_{X, v, K/2}\|_2 \sum_{i = 1}^{\left| C_{\tau}({\bf c}_{X,v})\right| - K} (1-\tau)^{i/2} \nonumber  \\
		& \leq & \sqrt{Q/\tau} (1 - \tau)^{K/4}(2/\tau) \|{\bf c}_{X, v, K/2}\|_2 \nonumber \\
		\tn{(By our setting of }K\tn{)}						& \leq & \tau^2 \|{\bf c}_{X, v, K/2}\|_2 \label{eqn:truncbound}
	\end{eqnarray}
	
	Given a choice of $\{S_j, S'_j, b_j\}_{j=1}^m$ we can define the following subset:
	\begin{align}
		A & := \{i \in [K/4] :  \pi(i) = j, b_j = 0, v \not\in S_j\}
	\end{align}
	Let $\mc{D}'$ be the restriction of $\mc{D}^e$ fixing everything else except for values of ${\bf X}_{v, i}$ for $i \in A \cup  \left(C_{\tau}({\bf c}_{X,v})\setminus [K]\right)$. 
	We have the following two claims which we shall prove later in this section.
	\begin{claim}\label{cl:trunc1}
		\begin{equation}
			\Pr_{\mc{D}'}\left[\left|\sum_{i \in A} \Langle {\bf c}_{X, v, i}, {\bf X}_{v,i}\Rangle + \Delta \right| \leq \|{\bf c}_{X, v, K/2}\|_2\right] \leq |A|^{-1/4}, \label{eqn:cltrunc1}
		\end{equation}
		for any constant $\Delta$.
	\end{claim}
	\begin{claim}\label{cl:trunc2}
		$$\Pr_{\mc{D}^e}\left[|A| \leq K\zeta/8\right] \leq \tn{exp}(-K\zeta/64).$$
	\end{claim}

	From the construction of $\tilde{\bf c}_{X, v}$ defined earlier,  we obtain that under $\mc{D}'$,
	\begin{eqnarray}
		& & \tilde{h}\left({\bf X},{\bf Y}\right) = \sum_{i \in A} \Langle {\bf c}_{X, v, i}, {\bf X}_{v,i} \Rangle + \theta', \nonumber \\
		& \Rightarrow & {h}\left({\bf X},{\bf Y}\right) = \sum_{i \in A} \Langle {\bf c}_{X, v, i}, {\bf X}_{v,i}\Rangle + \theta' + {h}\left({\bf X},{\bf Y}\right) - \tilde{h}\left({\bf X},{\bf Y}\right)
	\end{eqnarray} where $\theta'$ is some constant.
	The above implies that,
	\begin{eqnarray}
		\Pr_{\mc{D}'}\left[{\rm pos}\left(h\left({\bf X},{\bf Y}\right)\right) \neq {\rm pos}\left(\tilde{h}\left({\bf X},{\bf Y}\right)\right)\right] & \leq & \Pr_{\mc{D}'}\left[\left|\sum_{i \in A} \Langle {\bf c}_{X, v, i}, {\bf X}_{v,i} \Rangle + \theta' \right| \leq \left|h\left({\bf X},{\bf Y}\right) - \tilde{h}\left({\bf X},{\bf Y}\right)\right| \right] \nonumber \\
		\tn{(By \eqref{eqn:truncbound})} & \leq & \Pr_{\mc{D}'}\left[\left|\sum_{i \in A} \Langle {\bf c}_{X, v, i}, {\bf X}_{v,i} \Rangle + \theta' \right| \leq \tau^2 \|{\bf c}_{X, v, K/2}\|_2 \right] \nonumber \\
		\tn{(By Claim \ref{cl:trunc1})}  & \leq & |A|^{-1/4}.
	\end{eqnarray}
	Using the above along with Claim \ref{cl:trunc2} we can upper bound the LHS of \eqref{eqn:lemmatrucn} by 
	\begin{equation}
		\left(K\zeta/8\right)^{-1/4} + \tn{exp}(-K\zeta/64), \label{eqn:trunc2}
	\end{equation}
	which is at most $\tau^{1/4}/4k$ by our setting of parameters in \eqref{eqn:param} and large enough $z$ in Theorem \ref{thm:slc-hardness}.

	For the case when $v \in e_Y$, the proof is analogous to the above. The quantitative difference arises from defining $A$ instead as $\{i \in [K/4] :  \pi(i) = j, b_j = 1, v \not\in S'_j\}$ and replacement of $\zeta$ by $(1 -\zeta)$ in the corresponding version of Claim  \ref{cl:trunc2} and in \eqref{eqn:trunc2}. By our setting of the parameters, the LHS of \eqref{eqn:lemmatrucn} remains bounded by $\tau^{1/4}/4k$. 
\end{proof}

\begin{proof} (of Claim \ref{cl:trunc1}) Observe that by Proposition \ref{prop:crit-ix},  
	$$\|{\bf c}_{X, v, K/2}\|_2 \leq \sqrt{1/\tau}(1 - \tau)^{-K/4} \min_{i \in A} \|{\bf c}_{X, v, K/2}\|_2 \leq (4/K)\min_{i \in A} \|{\bf c}_{X, v, i}\|_2 \leq (1/|A|)\min_{i \in A} \|{\bf c}_{X, v, i}\|_2,$$
	where the penultimate inequality follows from the setting of $K$ which is large enough. Using this the LHS of \eqref{eqn:cltrunc1} can be upper bounded by,
	$$\Pr_{\mc{D}'}\left[\left|\sum_{i \in A} \Langle {\bf c}_{X, v, i}, {\bf X}_{v,i}\Rangle + \Delta \right| \leq \frac{\min_{i \in A} \|{\bf c}_{X, v, i}\|_2}{|A|}\right],$$
	which is at most $|A|^{-1/4}$ by an application of Lemma \ref{lem:block-LO}.
\end{proof}

\begin{proof} (of Claim \ref{cl:trunc2}) Since $e$ is nice w.r.t. $({\bf c}_X, {\bf c}_Y)$, we have that $|\pi([K/4])| = K/4$. Thus, each $i \in [K/4]$ is independently chosen to be in $A$ w.p. $\zeta(1 - t/k) \geq \zeta/2$. An application of Chernoff bound completes the proof.

\end{proof}

\section{Main Structural Lemma}					\label{sec:struct}
For the rest of this section we shall consider  a linear form given by 
$h({\bf X}, {\bf Y}) = \Langle {\bf c}_{X}, {\bf X}\Rangle + \Langle {\bf c}_{Y}, {\bf Y}\Rangle + \theta$, and an edge $e$ which is nice w.r.t $({\bf c}_X, {\bf c}_Y)$. Let $\pi_v := \pi_{e,v}$ for $v \in e$.
Further, we assume that $({\bf c}_X, {\bf c}_Y)$ satisfies
\begin{equation}
	C_\tau({\bf c}_{X, u}) = C^{\leq K}_\tau({\bf c}_{X, u}), \forall u \in e_X, \tn{\ \  and\ \ }  C_\tau({\bf c}_{Y, v}) = C^{\leq K}_\tau({\bf c}_{Y, v}), \forall v \in e_Y. \label{eqn:lemstructsetting}
\end{equation}
For convenience we define the following notation for each $v \in e$:
\begin{align}
	B_v &:= \begin{cases}
		C_\tau({\bf c}_{X, v}) & \tn{if } v \in e_X, \\
		C_\tau({\bf c}_{Y, v}) & \tn{if } v \in e_Y.
	\end{cases}
	& {\bf c}_{v} &:= \begin{cases}
		{\bf c}_{X, v} & \tn{if } v \in e_X, \\
		{\bf c}_{Y, v} & \tn{if } v \in e_Y.
	\end{cases}
	\nonumber \\
	 \beta_v &:= \begin{cases}
		0 & \tn{if } v \in e_X, \\
		1 & \tn{if } v \in e_Y.
	 \end{cases}
	 & \hat{S}_{v, j} &:= \begin{cases}
		 S_j & \tn{if } v \in e_X, \\
		 S'_j & \tn{if } v \in e_Y.
	 \end{cases}
	 {\tn \ \ } \forall j \in [m].
\end{align}
Further, let $P \subseteq [m]$ denote $\cup_{v\in e} \pi_v(B_v)$ and,
\begin{equation}
	{\bf c}^{(j)}_v :=  \left({\bf c}_{v, i}\right)_{i \in \pi_v^{-1}(j)}, \tn{\ \ \ \ }  {\bf c}^{\rm reg}_v :=  \left({\bf c}_{v, i}\right)_{i \in [M] \setminus (\pi_v^{-1}(P) \cup B_v)}, \tn{\ \ \ \ } {\bf c}^{\rm reg} =  ({\bf c}^{\rm reg}_v)_{v \in e}.
\end{equation}

We have the following lemma.

\begin{lemma}
\label{lem:main-struct}
	Given the above setting, if the following two conditions are satisfied,
	\begin{itemize}
		\item[1.] {\bf No Weak Intersections:} For every $u, v \in e$,   $\pi_{u}(B_u) \cap \pi_{v}(B_v) = \emptyset$,
		\item[2.] {\bf No Large Regular Top Blocks:} For every $v \in e, j \in P\setminus \pi_v(B_v$),  $\|{\bf c}^{(j)}_{v}\|^2 \leq \tau^4\sum_{i \notin B_v}\|\bc_{v,i}\|^2$, 
	\end{itemize}
	then the following holds,
	\begin{equation}
		\E_{\widehat{\mc{D}}}\left[\left|{\rm pos}\left(h\left({\bf X}^0,{\bf Y}^0\right)\right) - {\rm pos}\left(h\left({\bf X}^1,{\bf Y}^1\right)\right)\right|\right] \leq O(\tau). \label{eqn:lemmainstructeqn}
	\end{equation}
\end{lemma}

The rest of this section is devoted to proving the above lemma. To achieve a contradiction, we assume its two conditions.
First we make a couple of easily verifiable observations from condition 1 of the lemma. 

\begin{observation}				\label{obs:weak-index}
For every $j \in P$, there exists a unique vertex $v_j$ such that $B_{v_j} \cap \pi^{-1}_{v_j}(j) \neq \emptyset$.
\end{observation}

\begin{observation}				\label{fct:bound}
	The set $P$ satisfies $|P| \leq \sum_{v\in e}|B_v| \leq (2k)K$.
\end{observation}

Furthermore, condition $2$ of Lemma \ref{lem:main-struct} yields the following.

\begin{claim}						\label{cl:reg-1}
	For every vertex $v \in e$, $\|{\bf c}^{\rm reg}_v\|^2 \geq (1-\tau)\sum_{i \notin B_v} \|{\bf c}_{v,i}\|^2$.
\end{claim}
\begin{proof}
	For any vertex $v$, we have 
	\begin{equation}
		\sum_{j \in P}\sum_{i \in \pi^{-1}_v(j) \setminus B_v}  \|{\bf c}_{v,i}\|^2 = \sum_{j \in P\setminus \pi_v(B_v)}\sum_{i \in \pi^{-1}_v(j)} \|{\bf c}_{v,i}\|^2 + \sum_{j \in \pi_v(B_v)}\sum_{i \in \pi^{-1}_v(j) \setminus B_v} \|{\bf c}_{v,i}\|^2. \label{eqn:clreg-1}
	\end{equation}
The first term in the RHS of the above can be bounded by,
		\begin{equation}
			\leq \sum_{j \in P}\tau^4 \left(\sum_{i \notin B_v} \|{\bf c}_{v,i}\|^2\right) \leq \tau^4|P| \sum_{i \notin B_v}\|{\bf c}_{v,j}\|^2 \leq  (\tau/2) \sum_{i \in B_v}\|{\bf c}_{v,j}\|^2 
	\end{equation}
	For the second term, observe that by the niceness of $e$ w.r.t. $({\bf c}_X, {\bf c}_Y)$, for any $j \in \pi_v(B_v)$, we have $ |\pi^{-1}_v(j) \cap C^{\leq K}_\tau(\bc)| = 1$, and therefore for any other $i \in \pi^{-1}_v(j) \setminus B_v$ we must have that the value of $\|{\bf c}_{v,j}\|^2$ is at most $(1/d^8) \sum_{i \notin B_v} \|{\bf c}_{v,i}\|^2$. Since the number of such values in the summation is at most $d|B_v| \leq d K$, by setting $z$ in Theorem \ref{thm:slc-hardness} large enough, this summation can be bounded by $(\tau/2) \sum_{i \notin B_v}\|{\bf c}_{v,j}\|^2$. 
	Therefore, 
	\[
		\|{\bf c}^{\rm reg}_v\|^2 = \sum_{i \notin B_v} \|{\bf c}_{v,i}\|^2 - \sum_{j \in P} \sum_{i \in \pi^{-1}_v(j) \setminus B_v} \|{\bf c}_{v,j}\|^2 \geq (1-\tau)\sum_{j \notin B_v} \|{\bf c}_{v,j}\|^2.  
	\]
\end{proof}

\begin{claim}					\label{cl:reg-2}
	For every vertex $v \in e$, the coefficient vector ${\bf c}^{\rm reg}_v$ is $\tau'$-regular where $\tau' = \tau/(1 - \tau)$.
\end{claim}
\begin{proof}
	This follows from the fact that $\{{\bf c}_{v,i}\}_{i\not\in B_v}$ is $\tau$-regular, contains all of the ${\bf c}_{v,i}$ constituting ${\bf c}^{\rm reg}_v$, and from Claim \ref{cl:reg-1}.
\end{proof}

The following lemma provides a useful concentration for some $v \in e$ the sum of squared coefficients $\|{\bf c}_{v,i}\|^2$ corresponding to those indices $i$ constituting ${\bf c}^{\rm reg}_v$ s.t. the corresponding variables (${\bf X}_{v,i}$ or ${\bf Y}_{v,i}$ depending on whether $v$ is in $e_X$ or $e_Y$) are sampled u.a.r. from $\{0,1\}^Q$.  

\begin{lemma}				\label{lem:noisy-conc}
	For any vertex $v \in e$, over the choice of $\{b_j, S_j, S'_j\}_{j=1}^m$,
	\begin{equation}
		\Pr\left[\sum_{\substack{j \in [m]\setminus P: (b_j = \beta_v) \\ \wedge(v \not\in \hat{S}_{v,j})}} \ \sum_{i \in \pi_v^{-1}(j)} \|{\bf c}_{v,i}\|^2 \ \leq\ \frac{\zeta}{8}\|{\bf c}^{\rm reg}_v\|^2 \right] \leq \exp\left(-\frac{\zeta^2}{64\tau}\right), \label{eqn:lemnoisy}
	\end{equation}
	where 
\end{lemma}
\begin{proof}
	Define the random variable $\phi_j := \mathbbm{1}\{b_j = \beta_v, v \not\in \hat{S}_{v,j}\}\cdot\sum_{i \in \pi_v^{-1}(j)}	\|{\bf c}_{v,i}\|^2$ for each $j \in [m]\setminus P$. Note that $\{\phi_j\}_{j \in [m]\setminus P}$ are independent non-negative random variables and the summation inside the probability expression on the LHS of \eqref{eqn:lemnoisy} is precisely the random variable $\sum_{j \in [m]\setminus P}\phi_j$. Since $\{b_j = \beta_v, v \not\in \hat{S}_{v,j}\}$ occurs with probability $(1 - \zeta)(1 - t/k)$ if $v \in e_Y$ and $\zeta(1 - t/k)$ if $v \in e_X$, letting $\alpha_v := (1 - t/k)(\beta_v\zeta + (1 - \beta_v)(1 - \zeta))$ we obtain,
	\begin{equation}
		\E\left[\sum_{j \in [m]\setminus P}\phi_j\right] = \alpha_v\|{\bf c}^{\rm reg}_v\|^2. \label{eqn:lemnoisy0}
	\end{equation}
	Further, for each $j \in  [m]\setminus P$. 
	\begin{equation}
		(\max \phi_j)^2 = \left(\sum_{i \in \pi_v^{-1}(j)}	\|{\bf c}_{v,i}\|^2\right)^2 
			 = \sum_{i \in \pi_v^{-1}(j)} \|{\bf c}_{v,i}\|^4 + \sum_{\substack{i,i' \in \in \pi_v^{-1}(j) \\ i \neq i'}} \|{\bf c}_{v,i}\|^2 \|{\bf c}_{v,i'}\|^2 \label{eq:lemnoisy1}
	\end{equation}
	The first term on the RHS of \eqref{eq:lemnoisy1} can be upper bounded using the $\tau'$-regularity of ${\bf c}^{\rm reg}_v$ (Claim \ref{cl:reg-2}) as follows:
	\begin{equation}
		\sum_{i \in \pi_v^{-1}(j)} \|{\bf c}_{v,i}\|^4 \ \leq \ \tau'\|{\bf c}^{\rm reg}_v\|^2\sum_{i \in \pi_v^{-1}(j)} \|{\bf c}_{v,i}\|^2. \label{eq:lemnoisy2}
	\end{equation}
	On the other hand using the niceness of $e$ and Claim \ref{cl:reg-1}, we obtain that all $i \in \pi_v^{-1}(j)$ except for at most one satisfy  $\|{\bf c}_{v,i}\|^2 < (1/d^8)\sum_{i \notin B_v} \|{\bf c}_{v,i}\|^2  \leq (2/d^8) \|{\bf c}^{\rm reg}_v\|^2$. Using this, the second term on the RHS of  \eqref{eq:lemnoisy1} is at most,
	\begin{eqnarray}
		\sum_{\substack{i,i' \in \in \pi_v^{-1}(j) \\ i \neq i'}}\left( \max\{\|{\bf c}_{v,i}\|^2 , \|{\bf c}_{v,i'}\|^2\}\cdot\left(2/d^8\right) \|{\bf c}^{\rm reg}_v\|^2\right) & \leq & \frac{4}{d^8}\sum_{i \in \pi_v^{-1}(j)}\left[\|{\bf c}_{v,i}\|^2\sum_{\substack{i' \in \pi_v^{-1}(j) \\ i' \neq i}}\|{\bf c}^{\rm reg}_v\|^2\right] \nonumber \\ 
		&\leq & \left( \frac{4}{d^7}\right)\|{\bf c}^{\rm reg}_v\|^2\sum_{i \in \pi_v^{-1}(j)} \|{\bf c}_{v,i}\|^2.
	\end{eqnarray}
	Combining the above with \eqref{eq:lemnoisy2} and \eqref{eq:lemnoisy1} we obtain,
	\begin{equation}
		\sum_{j \in [m]\setminus P} (\max \phi_j)^2 \leq \left(\tau' + 4/d^7\right) \|{\bf c}^{\rm reg}_v\|^2 \sum_{i \in \pi_v^{-1}(j)} \|{\bf c}_{v,i}\|^2 \leq 3\tau \|{\bf c}^{\rm reg}_v\|^4.
	\end{equation}
	Using the above along with \eqref{eqn:lemnoisy0}, we apply the Chernoff-Hoeffding inequality (Theorem \ref{thm:chernoff}) to the sum  $\sum_{j \in [m]\setminus P}\phi_j$ as follows:
	\begin{equation}
		\Pr\left[\sum_{j \in [m]\setminus P}\phi_j \leq \frac{\alpha_v}{2}\|{\bf c}^{\rm reg}_v\|^2\right]  \leq  2 \tn{exp}\left(\frac{-\alpha_v^2 \|{\bf c}^{\rm reg}_v\|^4}{6\tau \|{\bf c}^{\rm reg}_v\|^4}\right) 
		 \leq  \exp\left(-\frac{\alpha_v^2}{12\tau}\right), 
	\end{equation}
	which using the fact that $\alpha_v > \zeta/2$ (since $\zeta < 1/2$ and $t < k/2$) proves the lemma.
\end{proof}

A repeated application of Lemma \ref{lem:noisy-conc} for each vertex  $v \in e$ along with a union bound yields,
\begin{equation}	
	\Pr\left[\sum_{v\in e}\ \sum_{\substack{j \in [m]\setminus P: (b_j = \beta_v) \\ \wedge(v \not\in \hat{S}_{v,j})}}\ \sum_{i \in \pi_v^{-1}(j)} \|{\bf c}_{v,i}\|^2 \ \leq\ \frac{\zeta}{8}\|{\bf c}^{\rm reg}\|^2 \right] \leq 2k\cdot\exp\left(-\frac{\zeta^2}{64\tau}\right) \label{eqn:noisyall}
\end{equation}
 
Let us define, 
\begin{eqnarray}
	\mc{S}\left({\bf c}_{X}, {\bf c}_{Y}, {\bf X}, {\bf Y}, \{b_j, S_j , S_j'\ :\ j \in [m]\}\right) :=  & &\sum_{v\in e_X}\ \sum_{\substack{j \in [m]\setminus P : v \not\in S_j}} \ \sum_{i \in \pi_v^{-1}(j)} \Langle {\bf c}_{X, v,i}, {\bf X}_{v,i} \Rangle \nonumber \\ & & +  \sum_{v\in e_X}\ \sum_{\substack{j \in [m]\setminus P : v \not\in S'_j}} \ \sum_{i \in \pi_v^{-1}(j)} \Langle {\bf c}_{Y, v,i}, {\bf Y}_{v,i} \Rangle \label{eqn:Lambda}
\end{eqnarray}
The above definition of $\mc{S}$ captures the contribution due to the ${\bf c}_{v, i}$ constituting ${\bf c}^{\rm reg}$ such that the variables (${\bf X}_{v,i}$ or ${\bf Y}_{v,i}$ depending on whether $v$ is in $e_X$ or $e_Y$) are sampled u.a.r. from $\{0,1\}^Q$. We prove the following anti-concentration of $\mc{S}$.
\begin{lemma}					\label{lem:anti-conc-2}
	\begin{eqnarray}
		\Pr_{\mc{D}^e} \left[ \left| \mc{S}\left({\bf c}_{X}, {\bf c}_{Y}, {\bf X}, {\bf Y}, \{b_j, S_j , S_j'\ :\ j \in [m]\}\right) + \theta'\right| \leq \eps_0 \right] \leq  O(\tau) & + & \frac{64\eps_0}{\|{\bf c}^{\rm reg}\|_2\sqrt{\zeta}} \nonumber \\&  & + 2k\cdot\exp\left(-\frac{\zeta^2}{64\tau}\right),
	\end{eqnarray}
	where $\eps_0 \geq 0$ is a constant, and after fixing $\{b_j, S_j, S_j'\}_{j\in [m]}$ $\theta'$ does not depend on the variables ${\bf X}_{v,i}$ and ${\bf Y}_{v,i}$ in \eqref{eqn:Lambda}.
\end{lemma}
\begin{proof}
	Define ${\bf Z}_v := {\bf X}_v - \frac{1}{2}{\bf 1}$ if $v \in e_X$ and ${\bf Y}_v - \frac{1}{2}{\bf 1}$ otherwise. Note that for any $v$, $j \in [m]\setminus P$ s.t. $b_j = \beta_v$ and $v \not\in \hat{S}_{v,j}$, and $i \in \pi_v^{-1}(j)$, the vector ${\bf Z}_{v,i}$ is uniformly sampled from $\{-1/2,1/2\}^Q$. Further we have,
	\begin{eqnarray}
		\mc{S}\left({\bf c}_{X}, {\bf c}_{Y}, {\bf X}, {\bf Y}, \{b_j, S_j , S_j'\}_{j \in [m]}\right) + \theta'
		& = & 
		\sum_{v\in e}\sum_{\substack{j \in [m]\setminus P: (b_j = \beta_v) \\ \wedge(v \not\in \hat{S}_{v,j})}} \ \sum_{i \in \pi_v^{-1}(j)} \Langle {\bf c}_{v,i}, {\bf Z}_{v,i} \Rangle \  + \Theta \label{eqn:anticoncshift} \\
		 & = & \sum_{v\in e}\sum_{\substack{j \in [m]\setminus P: (b_j = \beta_v) \\ \wedge(v \not\in \hat{S}_{v,j})}} \ \sum_{i \in \pi_v^{-1}(j)}\sum_{q\in[Q]}c_{v,i,q}Z_{v,i,q}\  +  \Theta, \label{eqn:anticoncshift2} 
	\end{eqnarray}
	where $\Theta$ is a fixed constant after fixing $\{b_j, S_j\}_{j\in [m]}$ along with any value of $\theta'$, and randomizing only on the values of the ${\bf Z}_{v,i}$ appearing on the RHS of \eqref{eqn:anticoncshift}. Any product $c_{v,i,q} Z_{v,i,q}$ on the RHS of \eqref{eqn:anticoncshift2} is a mean zero random variable with variance $c_{v,i,q}^2/4$ and third moment $|c_{v,i,q}^3|/8$.  The maximum over all the random variables $c_{v,i,q} Z_{v,i,q}$ of the ratio of its third moment and second moments is the maximum value of $|c_{v,i,q}|/2$ over $(v,i,q)$ appearing in  \eqref{eqn:anticoncshift2}. By Claim \ref{cl:reg-2} this is at most $\tau'\|{\bf c}^{\rm reg}\|$. Further, the sum of variances $\gamma$ is (by \eqref{eqn:noisyall}) at least $(\frac{\zeta}{32})\|{\bf c}^{\rm reg}\|^2$ except with probability  $2k\cdot\exp\left(-\zeta^2/(64\tau)\right)$ over the choice of $\{b_j, S_j, S'_j\}_{j\in [m]}$. Assuming this and applying the Berry-Esseen theorem (Theorem \ref{thm:berry-ess}) we obtain that the CDF $F$ of the RHS of \eqref{eqn:anticoncshift} excluding $\Theta$, satisfies $\left|F(x) - \Phi_\gamma(x)\right| \leq O(\tau') = O(\tau)$ for any $x \in (-\infty,\infty)$, where $\Phi_\gamma$ is the CDF of the normal distribution $N(0, \gamma)$. In this case, the LHS of \eqref{eqn:anticoncshift} lies in a fixed interval of length $2\eps_0$ with probability at most $O(\tau) + 16\eps_0/(\sqrt{\zeta} \|{\bf c}^{\rm reg}\|)$. Losing additional probability of $2k\cdot\exp\left(-\zeta^2/(64\tau)\right)$ for our assumption by \eqref{eqn:noisyall} completes the proof of Lemma \ref{lem:anti-conc-2}.
\end{proof}

\subsection{Proof of Lemma \ref{lem:main-struct}}

	Our first goal is to bound the variance of, $\left|h({\bf X}^1, {\bf Y}^1) - h({\bf X}^0, {\bf Y}^0)\right|$ under the distribution  $\widehat{\mc{D}}$. Let the difference for each $j \in [m]$ be represented by $\Delta_j$ so that $\Ex\left[h({\bf X}^1, {\bf Y}^1) - h({\bf X}^0, {\bf Y}^0)\right] = \E\left[\sum_j \Delta_j\right]$. Observe that $\E\left[\Delta_j\right] = 0$ due to matching expectations of every variable under $\widehat{\mc{D}}$. Explicitly,
\begin{eqnarray}
	\Delta_j & := & \sum_{v \in e_X}\sum_{i \in \pi_v^{-1}(j)\setminus B_v}\mathbbm{1}\{v = u_{X, j}\}\Langle {\bf c}_{X, v, i}, {\bf X}^1_{v, i}\Rangle + \sum_{v \in e_Y}\sum_{i \in \pi_v^{-1}(j)\setminus B_v}\mathbbm{1}\{v = u_{Y, j}\}\Langle {\bf c}_{Y, v, i}, {\bf Y}^1_{v, i} \Rangle   \nonumber \\
	& & - \sum_{v \in e_X}\sum_{i \in \pi_v^{-1}(j)\setminus B_v}\mathbbm{1}\{v \in T_j\}\Langle {\bf c}_{X, v, i}, {\bf X}^0_{v, i}\Rangle + \sum_{v \in e_Y}\sum_{i \in \pi_v^{-1}(j)\setminus B_v}\mathbbm{1}\{v \in T'_j\}\Langle {\bf c}_{Y, v, i}, {\bf Y}^0_{v, i}\Rangle.
\end{eqnarray}
Note that the ${\bf X}^1_{v, i}, {\bf X}^0_{v, i}, {\bf Y}^1_{v, i}$ or ${\bf Y}^0_{v, i}$ appearing in the above expression are independent random variables sampled u.a.r. from $\{{\bf e}_1, \dots, {\bf e}_Q\}$. Therefore, 
\begin{equation}
	\E\left[\Langle {\bf c}_{X, v, i}, {\bf X}^1_{v, i}\Rangle^2\right] = \|{\bf c}_{X, v, i}\|^2/Q, \label{eqn:eqprodexp}
\end{equation}
and similarly for $\Langle {\bf c}_{X, v, i}, {\bf X}^0_{v, i}\Rangle$, $\Langle {\bf c}_{Y, v, i}, {\bf Y}^1_{v, i}\Rangle$, and $\Langle {\bf c}_{Y, v, i}, {\bf Y}^0_{v, i}\Rangle$. 
There are at most $4kd$ inner products appearing in the above expression for $\Delta_j$, with at most two corresponding to each $(v, i)$ s.t. $v \in e, i \in   \pi_v^{-1}(j)\setminus B_v$. Since $\E\left[\Delta_j\right] = 0$, we have $\tn{Var}[\Delta_j] = \E[\Delta_j^2]$ which, by Cauchy-Schwartz and \eqref{eqn:eqprodexp} is at most,
$$\left(\sqrt{4kd}\right)\cdot2\sum_{v \in e}\sum_{i \in \pi_v^{-1}(j)\setminus B_v}\left(\frac{\|{\bf c}_{v, i}\|^2}{Q}\right) = \frac{1}{\sqrt{Q}}\sum_{v \in e}\sum_{i \in \pi_v^{-1}(j)\setminus B_v}\|{\bf c}_{v, i}\|^2,$$
where the inequality is implied by the setting of $Q$ in \eqref{eqn:param}.
Summing up over all $j \in [m]$ and using independence of $\{\Delta_j, j \in [m]\}$, we obtain that,
\begin{eqnarray}
	\tn{Var}\left[\left|h({\bf X}^1, {\bf Y}^1) - h({\bf X}^0, {\bf Y}^0)\right|\right] & \leq & \frac{1}{\sqrt{Q}}\sum_{j\in [m]}\sum_{v \in e}\sum_{i \in \pi_v^{-1}(j)\setminus B_v}\|{\bf c}_{v, i}\|^2 \nonumber \\
	& \leq & \frac{1}{\sqrt{Q}}\sum_{v \in e}\sum_{i \notin  B_v}\|{\bf c}_{v, i}\|^2 \nonumber \\
	& \leq & \frac{1}{\sqrt{Q}} (1 - \tau)^{-1}\sum_{v \in e}\|{\bf c}^{\rm reg}_v\|^2 \leq \frac{2\|{\bf c}^{\rm reg}\|^2}{\sqrt{Q}} \label{eqn:varbound0}
\end{eqnarray}
 where the second last inequality follows from applying Claim \ref{cl:reg-1} on each inner summation. Applying Chebyshev's inequality we obtain for any $\eps_0 > 0$,
	\begin{equation}
		\Pr_{\widehat{\mc{D}}}\left[\left|h({\bf X}^1, {\bf Y}^1) - h({\bf X}^0, {\bf Y}^0)\right| > \eps_0  \right] < \frac{2\|{\bf c}^{\rm reg}\|^2}{\eps_0^2\sqrt{Q}}. \label{eqn:cheby}
	\end{equation}

To complete the proof of Lemma \ref{lem:main-struct}, observe that under $\widehat{\mc{D}}$, 
$${\bf X}^1_{v, i} = {\bf X}^0_{v, i} =: {\bf X}_{v, i} \overset{\rm u.a.r}{\sim} \{0,1\}^Q, \forall v \in e_X, j \in [m]\setminus P : b_j = 0, v \not\in S_j, i \in \pi_v^{-1}(j),$$
$${\bf Y}^1_{v, i} = {\bf Y}^0_{v, i} =: {\bf Y}_{v, i} \overset{\rm u.a.r}{\sim} \{0,1\}^Q, \forall v \in e_Y, j \in [m]\setminus P : b_j = 1, v \not\in S'_j, i \in \pi_v^{-1}(j).$$
Thus, under  $\widehat{\mc{D}}$
\begin{eqnarray}
	& & h({\bf X}^0, {\bf Y}^0) = \mc{S}\left({\bf c}_{X}, {\bf c}_{Y}, {\bf X}, {\bf Y}, \{b_j, S_j , S_j'\ :\ j \in [m]\}\right) + \theta'  \label{eqn:arrange1} \\
	&\Rightarrow & h({\bf X}^1, {\bf Y}^1) = \mc{S}\left({\bf c}_{X}, {\bf c}_{Y}, {\bf X}, {\bf Y}, \{b_j, S_j , S_j'\ :\ j \in [m]\}\right) + \theta' + h({\bf X}^1, {\bf Y}^1) -  h({\bf X}^0, {\bf Y}^0), \label{eqn:arrange2}
\end{eqnarray}
where $\theta'$ does not depend on the variables ${\bf X}_{v,i}$ and ${\bf Y}_{v,i}$ appearing in the expression \eqref{eqn:Lambda} for $\mc{S}$ after fixing $\{b_j, S_j, S_j'\}_{j\in [m]}$. The equations \eqref{eqn:arrange1} and \eqref{eqn:arrange2} also imply that (for some constant $\eps_0$ that we shall choose shortly),
\begin{eqnarray}
	& & \Pr_{\widehat{\mc{D}}}\left[{\rm pos}(h({\bf X}^1, {\bf Y}^1)) \neq {\rm pos}(h({\bf X}^0, {\bf Y}^0))\right] \nonumber \\
	& \leq & \Pr_{\widehat{\mc{D}}}\left[\left|\mc{S}\left({\bf c}_{X}, {\bf c}_{Y}, {\bf X}, {\bf Y}, \{b_j, S_j , S_j'\}_{j \in [m]}\right) + \theta' \right| \leq \left|h({\bf X}^1, {\bf Y}^1) -  h({\bf X}^0, {\bf Y}^0)\right|\right] \nonumber \\
	&\leq & \Pr_{\widehat{\mc{D}}}\left[\neg \left(\left(\left|\mc{S}\left({\bf c}_{X}, {\bf c}_{Y}, {\bf X}, {\bf Y}, \{b_j, S_j , S_j\}_{j \in [m]}\right) + \theta' \right| \geq \eps_0\right)\bigwedge \left( \left|h({\bf X}^1, {\bf Y}^1) -  h({\bf X}^0, {\bf Y}^0)\right| \leq \eps_0\right)\right)\right] \nonumber \\ 
	& \leq & \Pr_{\widehat{\mc{D}}}\left[\left(\left|\mc{S}\left({\bf c}_{X}, {\bf c}_{Y}, {\bf X}, {\bf Y}, \{b_j, S_j , S_j'\}_{j \in [m]}\right) + \theta' \right| \leq \eps_0\right)\bigvee \left( \left|h({\bf X}^1, {\bf Y}^1) -  h({\bf X}^0, {\bf Y}^0)\right| \geq\eps_0\right)\right] \nonumber \\ 
	& \leq & \Pr_{\widehat{\mc{D}}}\left[\left|\mc{S}\left({\bf c}_{X}, {\bf c}_{Y}, {\bf X}, {\bf Y}, \{b_j, S_j , S_j'\}_{j \in [m]}\right) + \theta' \right| \leq \eps_0\right] + \Pr_{\widehat{\mc{D}}}\left[\left|h({\bf X}^1, {\bf Y}^1) -  h({\bf X}^0, {\bf Y}^0)\right| \geq\eps_0\right]
\end{eqnarray}
Choosing $\eps_0$ to be $\tau\sqrt{\zeta}\|{\bf c}^{\rm reg}\|/64$, and applying  Lemma \ref{lem:anti-conc-2} and \eqref{eqn:cheby} we obtain that the LHS of \eqref{eqn:lemmainstructeqn} is bounded by,
$$O(\tau) + 2k\cdot\exp\left(-\frac{\zeta^2}{64\tau}\right) + O\left(\left(\tau^2\zeta\sqrt{Q}\right)^{-1}\right),$$
which by our setting of the parameters in \eqref{eqn:param} is  $O(\tau)$ for a large enough $z$ in Theorem \ref{thm:slc-hardness}, completing the proof.

\bibliographystyle{alpha}
\bibliography{and-bib}

\appendix

\section{Blockwise Small Ball Probability}			\label{sec:LO}

Here we prove the following extension of the Littlewood-Offord-Erd\H os lemma. 

\begin{lemma}			\label{lem:block-LO}
	Let $\bc_1,\bc_2,\ldots,\bc_T \in \mathbbm{R}^Q$ be such that $\|\bc_1\| \geq \|\bc_2 \| \geq \cdots \geq \|\bc_T\|$. Furthermore, let $\bX_1,\bX_2,\ldots,\bX_T$ be $Q$-dimensional Bernoulli vector random variables i.e.,  $\bX_i \overset{\rm u.a.r}{\sim}  \{0,1\}^Q$ for every $i \in [T]$. Then
		\[
		\sup_{\theta \in \mathbbm{R}}\Pr\left[\left|\sum_{i \in [T]} \langle \bc_i,\bX_i \rangle + \theta \right| \leq \frac{\|\bc_T\|}{T^{1/2}}\right] \leq O(T^{-1/2})
		\] 
\end{lemma}
\begin{proof}
	Let $\eta > 0$ be a quantity which is fixed later. For any $i \in [T]$, we write $\bc_{i} = \left(c_{i,q}\right)_{r \in [Q]}$. We consider the following two cases. 
	
	{\bf Case (i)}\[ \min_{i\in [T]}\max_{q \in [Q]} |c_{i,q}| \leq \eta\|\bc_T\|\] Then there exists $i^* \in [T]$ such that for every choice of $q \in [Q]$, we have
	 $|c_{i^*,q}| \leq \eta\|c_K\| \leq \eta \|\bc_{i^*}\|$. For every $q \in [Q]$, define random variable $Z_q = c_{i^*,q}\left(X_{i^*,q}- \frac12 \right)$, and let $\sigma^2_q := \Ex Z^2_{q}$. Furthermore, define $\sigma^2 := \sum_{q \in [Q]} \sigma^2_r$. It is easy to verify that $\sigma^2_r = c^2_{i^*,q}/2$ and $\sigma = \|\bc_{i^*}\|/\sqrt{2}$.
	
	Again for every choice of $q \in [Q]$, define $\tilde{Z}_q = Z_q/\sigma$. Then by construction we have (i) $\Ex[Z_q] = 0$ for every $q \in [Q]$ (ii) $\Ex\sum_{q \in [Q]} \tilde{Z}^2_q= 1$ and (iii) $\Ex \sum_{q \in [Q]} |\tilde{Z}_q|^3 \leq \eta$. Therefore, using the Berry Esseen Theorem (Theorem \ref{thm:berry-ess}), for any interval $I \subset \mathbbm{R}$ we have 
	\[
	\Pr\left[\left|\sum_{{q \in [Q]}}\tilde{Z}_q \right| \in I\right] \leq O(\eta) + \Pr_{g \sim N(0,1)} \Big[g \in I\Big] \leq O(\eta) + |I|
	\]
	Rolling back the sequence of transformations, for any fixing of $\left(\bX_j\right)_{j \neq i^*}$ and any choice of $\theta \in \mathbbm{R}$, we get that 
	\[
	\Pr_{\bX_{i^*}}\left[\left| \sum_{i \neq i^*} \langle \bc_i, \bX_i \rangle + \langle \bc_{i^*}, \bX_{i^*} \rangle + \theta \right| \leq  \eta\|\bc_T\|\right] \leq O(\eta) + \frac{\eta\|\bc_T\|}{\|\bc_{i^*}\|} \leq O(\eta)
	\]
	which gives us the bound for this case.
	
	{\bf Case (ii)}\[ \min_{i\in [T]}\max_{q \in [Q]} |c_{i,r}| > \eta\|\bc_K\|\] Here, for every choice of $i \in [T]$, there exists $q_i \in [Q]$ such that $|c_{i,q_i}| \geq \eta\|\bc_{T}\|$. Then using Littlewood-Offord-Erd\H os Lemma (Lemma \ref{lem:LO}), for any $\theta \in \mathbbm{R}$ we get that 
	\begin{eqnarray*}
	\Pr_{(X_{i,q_i})^{T}_{i = 1}}\left[\left|\sum_{i \in [T]} c_{i,q_i}X_{i,q_i} + \theta \right| \leq \eta\|\bc_{T}\| \right]
	&=& \Pr_{(X_{i,q_i})^{T}_{i = 1}}\left[\left|\sum_{i \in [T]} \left(\frac{c_{i,q_i}}{\eta\|\bc_T\|}\right)X_{i,q_i} + \frac{\theta}{\eta\|\bc_T\|} \right| \leq 1 \right] \\
	&\leq& O(T^{-1/2})
	\end{eqnarray*} 
	Since the above bound holds independent of the realization of $\left(X_{i,\neq r(i)}\right)^{T}_{i = 1}$ we have 
	\begin{eqnarray*}
	\Pr_{(\bX_{i})^{T}_{i = 1}}\left[\left|\sum_{i \in [T]} \langle \bc_{i},\bX_{i}\rangle + \theta \right| \leq \eta\|\bc_{T}\| \right] \leq O(T^{-1/2})
	\end{eqnarray*} 
	for any fixed choice of $\theta \in \mathbbm{R}$. Combining the two cases and setting $\eta = 1/\sqrt{T}$ completes the proof of the lemma.
	\end{proof}

\end{document}